\PassOptionsToPackage{table}{xcolor}
\documentclass[sigconf, nonacm]{acmart}

\usepackage{enumerate}
\usepackage{booktabs}
\usepackage{epsfig}
\usepackage{graphicx}
\usepackage{xspace}
\usepackage{float}
\usepackage{subfigure,amsmath,latexsym}
\usepackage{tikz}
\usepackage{url}
\usepackage{color}
\usepackage{bbm}
\usepackage{caption}
\usepackage{xcolor,colortbl}
\usepackage{multirow}
\usepackage{balance}

\newcommand{\cal}{\mathcal}

\usepackage[
resetcount,
algoruled,
linesnumbered,
vlined
]{algorithm2e}

\SetKw{OR}{or}
\SetKw{AND}{and}
\SetKw{NOT}{not}

\SetKw{TRUE}{true}
\SetKw{FALSE}{false}
\SetKw{NULL}{nil}

\SetKw{KwDownto}{downto}

\SetKwIF{gIf}{gElsIf}{gElse}{if}{then}{else if}{else}{end if}%
\SetKwIF{If}{ElseIf}{Else}{if}{then}{else if}{else}{end if}%
\SetKwFor{For}{for}{do}{end for}%
\SetKwFor{ForPar}{for}{do in parallel}{end for}%
\SetKwFor{ForEach}{for each}{do}{end for}%
\SetKwFor{ForAll}{forall}{do}{end for}%
\SetKwFor{While}{while}{do}{end while}%
\SetKw{Break}{break\;}%

\newcommand{\State}[1]{#1\;}
\newcommand{\StateCmt}[2]{
  #1; \tcc*[f]{#2}\;
}

\SetKwInOut{Input}{Input}
\SetKwInOut{Output}{Output}
\SetKwInput{KwDescription}{Description}

\renewcommand{\Return}[1]{\State{\textbf{return} #1}}

\let\oldnl\nl
\newcommand{\nonl}{\renewcommand{\nl}{\let\nl\oldnl}}

\newcommand{\stitle}[1]{\vspace{1ex} \noindent{\bf #1}}

\newcommand{\kw}[1]{{\ensuremath {\mathsf{#1}}}\xspace}

\newcommand{\la}{\leftarrow}

\newcommand{\cfig}{Figure~}
\newcommand{\ctab}{Table~}
\newcommand{\csec}{Section~}

\newcommand{\calg}{Algorithm~}
\newcommand{\cthm}{Theorem~}
\newcommand{\clem}{Lemma~}

\newcommand{\cequ}[1]{Equation~(#1)}

\newcommand{\ie}{{i.e.}}
\newcommand{\eg}{{e.g.}}

\newcommand{\bigo}{ {\cal O}}

\newcommand{\kdbb}{\kw{KDBB}}
\newcommand{\madec}{\kw{MADEC^+}}

\newcommand{\bbsearch}{\kw{Branch\&Bound}}

\newcommand{\lb}{lb}

\newcommand{\gT}{\cal T}
\newcommand{\gL}{\cal L}

\newcommand{\pivotplus}{\kw{Pivot\text{+}}}
\newcommand{\kdc}{\kw{kDC}}

\newcommand{\kdct}{\kw{kDC\text{-}two}}

\newcommand*\circled[1]{\tikz[baseline=(char.base)]{
            \node[fill=gray!10,shape=circle,draw,inner sep=0.5pt] (char) {#1};}}

\def\thepage{\arabic{page}}
\AtBeginDocument{%
  }

\newcommand\vldbdoi{XX.XX/XXX.XX}
\newcommand\vldbpages{XXX-XXX}
\newcommand\vldbvolume{14}
\newcommand\vldbissue{1}
\newcommand\vldbyear{2020}
\newcommand\vldbauthors{\authors}
\newcommand\vldbtitle{\shorttitle} 
\newcommand\vldbavailabilityurl{URL_TO_YOUR_ARTIFACTS}
\newcommand\vldbpagestyle{plain} 

\begin{document}

\title{Maximum Defective Clique Computation: Improved Time Complexities
and Practical Performance}

\author{Lijun Chang}
\affiliation{
    \institution{The University of Sydney}
	\city{Sydney}
    \country{Australia}
}
\email{Lijun.Chang@sydney.edu.au}

\begin{abstract}
The concept of $k$-defective clique, a relaxation of clique by
allowing up-to $k$ missing edges, has been receiving increasing
interests recently. 
Although the problem of finding the maximum $k$-defective clique is
NP-hard, several practical algorithms have been recently proposed in the
literature, with \kdc being the state of the art. \kdc not only runs the
fastest in practice, but also achieves the best time complexity.
Specifically, it runs in $\bigo^*(\gamma_k^n)$ time when ignoring
polynomial factors; here, $\gamma_k$ is a constant that is smaller than
two and only depends on $k$, and $n$ is the number of vertices in the
input graph $G$.
In this paper, we propose the \kdct algorithm to improve the
time complexity as well as practical performance.
\kdct runs in $\bigo^*( (\alpha\Delta)^{k+2}
\gamma_{k-1}^\alpha)$ time when the maximum $k$-defective clique size
$\omega_k(G)$ is at least $k+2$, and in $\bigo^*(\gamma_{k-1}^n)$ time
otherwise, where $\alpha$ and $\Delta$ are the degeneracy and maximum
degree of $G$, respectively.
Note that, most real graphs satisfy $\omega_k(G) \geq k+2$, and for
these graphs, we not only improve the base (\ie, $\gamma_{k-1} <
\gamma_k$), but also the exponent, of the exponential time complexity.
In addition, with slight modification, \kdct also runs in
$\bigo^*( (\alpha\Delta)^{k+2} (k+1)^{\alpha+k+1-\omega_k(G)})$ time by
using the degeneracy gap $\alpha+k+1-\omega_k(G)$ parameterization;
this is better than $\bigo^*(
(\alpha\Delta)^{k+2}\gamma_{k-1}^\alpha)$ when $\omega_k(G)$ is close to
the degeneracy-based upper bound $\alpha+k+1$.
Finally, to further improve the practical performance, we propose a new
degree-sequence-based reduction rule that can be efficiently applied,
and theoretically demonstrate its effectiveness compared with those
proposed in the literature.
Extensive empirical studies on three benchmark graph
collections, containing $290$ graphs in total, show that our
\kdct algorithm outperforms the existing fastest algorithm \kdc by
several orders of magnitude.
\end{abstract}

\maketitle  \pagestyle{plain}

\pagestyle{\vldbpagestyle}
\begingroup\small\noindent\raggedright\textbf{PVLDB Reference Format:}\\
\vldbauthors. \vldbtitle. PVLDB, \vldbvolume(\vldbissue): \vldbpages, \vldbyear.\\
\href{https://doi.org/\vldbdoi}{doi:\vldbdoi}
\endgroup
\begingroup
\renewcommand\thefootnote{}\footnote{\noindent
This work is licensed under the Creative Commons BY-NC-ND 4.0 International License. Visit \url{https://creativecommons.org/licenses/by-nc-nd/4.0/} to view a copy of this license. For any use beyond those covered by this license, obtain permission by emailing \href{mailto:info@vldb.org}{info@vldb.org}. Copyright is held by the owner/author(s). Publication rights licensed to the VLDB Endowment. \\
\raggedright Proceedings of the VLDB Endowment, Vol. \vldbvolume, No. \vldbissue\ %
ISSN 2150-8097. \\
\href{https://doi.org/\vldbdoi}{doi:\vldbdoi} \\
}\addtocounter{footnote}{-1}\endgroup

\ifdefempty{\vldbavailabilityurl}{}{
\vspace{.3cm}
\begingroup\small\noindent\raggedright\textbf{PVLDB Artifact Availability:}\\
The source code, data, and/or other artifacts have been made available at \url{\vldbavailabilityurl}.
\endgroup
}

\section{Introduction}
\label{sec:introduction}

Graphs have been widely used to capture the relationship between
entities in applications such as social media, communication network,
e-commerce, and cybersecurity. Identifying dense subgraphs from those
real-world graphs, which are usually globally sparse (\eg, have a small
average degree), is a fundamental problem and has received a lot of
attention~\cite{Book18:Chang,MMGD10:Lee}.
Dense subgraphs may correspond to communities in social
networks~\cite{bedi2016community}, protein complexes in biological
networks~\cite{suratanee2014characterizing}, and anomalies in financial
networks~\cite{ahmed2016survey}.
The clique model, requiring every pair of vertices to be
directly connected by an edge, represents the densest structure that a
subgraph can be. As a result, clique related problems have been
extensively studied, \eg, theoretical aspect of maximum clique
computation~\cite{SJC77:Tarjan,TC86:Jian,JA86:Robson,clique}, practical
aspect of maximum clique
computation~\cite{KDD19:Chang,WALCOM17:Tomita,ORL90:Carraghan,COR17:Li,ICTAI13:Li,JGO94:Pardalos,IM15:Pattabiraman,JSC15:Rossi,COR16:Segundo,WALCOM10:Tomita,ICDE13:Xiang},
maximal clique enumeration~\cite{JEA13:Eppstein,TODS11:Cheng},
and $k$-clique counting and enumeration~\cite{PVLDB20:Li,WSDM20:Jain}.

Requiring every pair of vertices to be explicitly connected by an edge
however is often too restrictive in practice, by noticing that data
may be noisy or incomplete and/or the data collection process may
introduce errors~\cite{EOR13:Pattillo}.
In view of this, various clique relaxation models have been formulated
and studied in the literature, such as quasi-clique~\cite{LATIN02:Abello},
plex~\cite{OR11:Balasundaram}, club~\cite{EJOR02:Bourjolly}, and
defective clique~\cite{Bio06:Yu}.
In this paper, we focus on the defective clique model, where defective
cliques have been used for predicting missing interactions between proteins
in biological networks~\cite{Bio06:Yu}, cluster
detection~\cite{MP22:Stozhkov,SIGMOD23:Dai}, transportation
science~\cite{TS02:Sherali}, and social network
analysis~\cite{WWW20:Jain,IJC21:Gschwind}.
A subgraph with $a$ vertices is a $k$-defective clique if it has at
least ${a \choose 2}-k$ edges, \ie, it misses at most $k$ edges from
being a clique. A $k$-defective clique is usually referred to by its
vertices, since maximal $k$-defective cliques are vertex-induced subgraphs.
Consider the graph in \cfig\ref{fig:introduction}, $\{v_1,\ldots,v_4\}$
is a maximum clique, $\{v_1,\ldots,v_4,v_7\}$ is a {\em maximal} $2$-defective
clique, and $\{v_1,\ldots,v_6\}$ is a {\em maximum} $2$-defective clique that
maximizes the number of vertices.

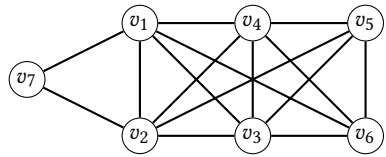
\begin{figure}[htb]
\centering
\begin{tikzpicture}[scale=0.5]
\node[draw, circle, inner sep=1.5pt] (1) at (0,0) {$v_1$};
\node[draw, circle, inner sep=1.5pt] (2) at (0,-3) {$v_2$};
\node[draw, circle, inner sep=1.5pt] (3) at (3,-3) {$v_3$};
\node[draw, circle, inner sep=1.5pt] (4) at (6,-3) {$v_6$};
\node[draw, circle, inner sep=1.5pt] (5) at (6,0) {$v_5$};
\node[draw, circle, inner sep=1.5pt] (6) at (3,0) {$v_4$};
\node[draw, circle, inner sep=1.5pt] (7) at (-3,-1.5) {$v_7$};
\path[draw,thick] (1) -- (2) -- (3) -- (4) -- (5) -- (6) -- (1);
\path[draw,thick] (6) -- (4) -- (1) -- (3) -- (6) -- (2) -- (5) -- (3);
\path[draw,thick] (1) -- (7) -- (2);
\end{tikzpicture}
\caption{Defective clique}
\label{fig:introduction}
\end{figure}

The state-of-the-art time complexity for maximum $k$-defective
computation is achieved by the \kdc algorithm proposed
in~\cite{SIGMOD24:Chang}, which runs in $\bigo^*(\gamma_k^n)$ time where
$\gamma_k < 2$ is the largest real root of the equation $x^{k+3} -
2x^{k+2}+1=0$ and $n$ is the number of vertices in the input graph.
The main ideas of achieving the time complexity are
$\circled{1}$~deterministically processing vertices that have
		up-to one non-neighbor (by reduction rule {\bf RR2}
		of~\cite{SIGMOD24:Chang}), and
$\circled{2}$~greedily ordering vertices (by branching rule
			{\bf BR} of~\cite{SIGMOD24:Chang}) such that the
			\mbox{length-$(k+2)$} prefix of the ordering induces more than $k$
			missing edges.
$\circled{2}$ ensures the time complexity since we only need to consider
up-to $k+2$ prefixes when enumerating the prefixes that can be added to
the solution, while $\circled{1}$ makes $\circled{2}$ possible.
Observing that $\circled{2}$ is impossible when each vertex has exactly one
non-neighbor,
it is natural to wonder whether the time complexity will be reduced if
there are techniques to make each vertex have more (than two)
non-neighbors. Unfortunately, the answer is negative.
An alternative way is to design a
different strategy (\ie, branching rule) for finding a subset of $k+1$ or fewer vertices that
induce more than $k$ missing edges.
However, this most likely cannot be conducted efficiently. Details of
these negative results will be
discussed in \csec\ref{sec:challenge}.

In this paper, we propose the \kdct algorithm to improve both the time
complexity and the practical performance for exact maximum $k$-defective
clique computation. 
Firstly, \kdct uses the same branching rule {\bf BR} and reduction rules
{\bf RR1} and {\bf RR2} as \kdc, but
we prove {\em a reduced base} (\ie, $\gamma_{k-1}$) for the time
complexity by using
different analysis techniques for different backtracking instances.
Specifically, let $\gT$ be the backtracking search tree (see
\cfig\ref{fig:full_tree}) where each node
represents a backtracking instance $(g,S)$ with $g$ being a subgraph of
the input graph $G$ and $S\subseteq V(g)$ a $k$-defective clique,
it is sufficient to bound the number of leaf nodes of $\gT$. 
Our general idea is that if at least one branching vertex (\ie,
previously selected by {\bf BR}) has been added to
$S$, then {\bf BR} computes an ordering of $V(g)\setminus S$ such that the
union of $S$ and a length-$(k+1)$ prefix of the ordering induce more
than $k$ missing edges (\clem\ref{lemma:bound_q}); consequently, the
number of leaf nodes of $\gT$
rooted at $(g,S)$ can be shown {\em by induction} to be at most
$\gamma_{k-1}^{|V(g)\setminus S|}$
(\clem\ref{lemma:leafs_of_nonleaf_nodes}).
Otherwise, we prove {\em non-inductively} that the number of leaf nodes
is at most $2\cdot \gamma_{k-1}^{|V(g)\setminus S|}$ by introducing the coefficient $2$ (\clem\ref{lemma:leafs}).

Secondly, \kdct also makes use of the diameter-two property of large
$k$-defective cliques (\ie, any $k$-defective clique of size $\geq k+2$
has a diameter at most two) to {\em reduce the exponent} of the time
complexity when $\omega_k(G) \geq k+2$. That is, once a vertex $u$ is
added to $S$, we can remove from $g$ those vertices whose shortest
distances (computed in $g$) to $u$ are larger than two.
Let $(v_1,v_2,\ldots,v_n)$ be a degeneracy ordering of $V(G)$. We
process each vertex $v_i$ by assuming that it is the first vertex of
the degeneracy ordering that is in the maximum $k$-defective clique;
note that, at least one of these $n$ assumptions will be true, and thus
we can find the maximum $k$-defective clique.
The search tree of processing $v_i$ has at most $2\cdot
\gamma_{k-1}^{\alpha\Delta}$ leaf nodes since we only need to consider
$v_i$'s neighbors and two-hop neighbors that come later than $v_i$ in
the degeneracy ordering; here $\alpha$ and $\Delta$ are
the degeneracy and maximum degree of $G$, respectively. Through a more
refined analysis, we show that the number of leaf nodes is also bounded
by $\bigo( (\alpha\Delta)^k \gamma_{k-1}^\alpha)$. Consequently, \kdct
runs in $\bigo(n\times (\alpha\Delta)^{k+2}\times
\gamma_{k-1}^{\alpha})$ time when $\omega_k(G) \geq k+2$; note that,
$\alpha \leq \sqrt{m}$ and is small in practice~\cite{JEA13:Eppstein},
where $m$ is the number of edges in $G$.
Furthermore, we show that \kdct, with slight modification, runs in
$\bigo^*( (\alpha\Delta)^{k+2} \times (k+1)^{\alpha+k+1-\omega_k(G)})$
time when using the degeneracy gap $\alpha+k+1-\omega_k(G)$
parameterization; this is better than
$\bigo^*(\gamma_{k-1}^\alpha)$ when $\omega_k(G)$ is close to its upper
bound $\alpha+k+1$.

Thirdly, we propose a new reduction rule {\bf RR3} to further improve the
practical performance of \kdct. {\bf RR3} is designed based on the
degree-sequence-based upper bound {\bf UB} proposed
in~\cite{AAAI22:Gao}. However, instead of using {\bf UB} to prune
instances after generating them as done in the existing
works~\cite{AAAI22:Gao,SIGMOD24:Chang}, we propose to remove vertex $u
\in V(g)\setminus S$ from $g$
if an upper bound of $(g, S\cup u)$ is no larger than $\lb$.
Note that, rather than computing the exact upper bound for $(g,S\cup
u)$, we test whether the upper bound is larger than $\lb$
or not. The latter can be conducted more efficiently and without generating
$(g,S\cup u)$; moreover, computation can be shared
between the testing for different vertices of $V(g)\setminus S$.
We show that with linear time preprocessing, the upper
bound testing for all vertices $u \in V(g)\setminus S$ can be conducted
in totally linear time.
In addition, we theoretically demonstrate the effectiveness of {\bf RR3}
compared with the existing reduction rules, \eg, the
degree-sequence-based reduction
rule and second-order reduction rule proposed
in~\cite{SIGMOD24:Chang}.

\stitle{Contributions.}
Our main contributions are as follows.
\begin{itemize}
	\item We propose the \kdct algorithm for exact maximum $k$-defective
		clique computation, and prove that it runs in
		$\bigo^*( (\alpha\Delta)^{k+2}\times \gamma_{k-1}^\alpha)$ time
		on graphs with $\omega_k(G)\geq k+2$ and in
		$\bigo^*(\gamma_{k-1}^n)$ time otherwise. This improves the
		state-of-the-art time complexity $\bigo^*(\gamma_k^n)$ by noting
		that $\gamma_{k-1} < \gamma_k$.
	\item We prove that \kdct, with slight modification, runs in
		$\bigo^*( (\alpha\Delta)^{k+2} \times
		(k+1)^{\alpha+k+1-\omega_k(G)})$ time when using the degeneracy
		gap $\alpha+k+1-\omega_k(G)$ parameterization.
	\item We propose a new degree-sequence-based reduction rule {\bf
		RR3} that can be conducted in linear time, and theoretically
		demonstrate its effectiveness compared with the
		existing reduction rules.
\end{itemize}
We conduct extensive empirical studies on three benchmark collections
with $290$ graphs in total to evaluate our techniques.
The results show that (1)~our algorithm \kdct solves $16$, $11$ and $7$
more graph instances than the fastest existing algorithm \kdc on the
three graph collections, respectively, for a time limit of $3$ hours and
$k=15$;
(2)~our algorithm \kdct solves all $114$ Facebook graphs with a time
limit of $30$ seconds for $k=1$, $3$ and $5$; (3)~on the $39$ Facebook
graphs that have more than $15,000$ vertices, \kdct is on average two
orders of magnitude faster than \kdc for $k=15$.

\stitle{Organizations.}
The remainder of the paper is organized as follows.
\csec\ref{sec:preliminaries} defines the problem, and
\csec\ref{sec:existing} reviews the state-of-the-art algorithm \kdc and its
time complexity analysis. We present our algorithm \kdct and its time
complexity analysis in \csec\ref{sec:approach}, and our new reduction
rule {\bf RR3} in \csec\ref{sec:rr}. Experimental results are discussed
in \csec\ref{sec:experiment}, followed by related works in
\csec\ref{sec:related_work}. Finally, \csec\ref{sec:conclusion}
concludes the paper.

\section{Problem Definition}
\label{sec:preliminaries}

We consider a large {\em unweighted}, {\em undirected} and
{\em simple} graph $G=(V, E)$ and refer to it simply as a graph; here,
$V$ is the vertex set and $E$ is the edge set.
The numbers of vertices and edges of $G$ are denoted by $n = |V|$ and
$m=|E|$, respectively.
An undirected edge between $u$ and $v$ is denoted by $(u,v)$ and $(v,u)$.
The set of edges that are missing from $G$ is called the set of
\textbf{\em non-edges} (or missing edges) of $G$ and denoted by $\overline{E}$, \ie,
$(u,v) \in \overline{E}$ if $u \neq v$ and $(u,v) \notin E$.
The set of $u$'s neighbors in $G$ is denoted $N_G(u) = \{v \in V \mid
(u,v) \in E\}$, and the {\em degree} of $u$ in $G$ is $d_G(u) =
|N_G(u)|$;
similarly, the set of $u$'s \textbf{\em non-neighbors} in $G$ is denoted
$\overline{N}_G(u) = \{v \in V \mid (u,v) \in \overline{E}\}$. Note that
\textbf{\em a vertex is neither a neighbor nor a non-neighbor of itself}.
Given a vertex subset $S \subseteq V$, the set of edges {\em induced} by $S$
is $E_G(S) = \{(u,v) \in E \mid u,v \in S\}$, the set of non-edges
induced by $S$ is $\overline{E}_G(S) = \{(u,v) \in \overline{E} \mid
u,v\in S\}$, and the subgraph of $G$ induced by $S$ is $G[S] = (S, E(S))$.
We denote the union of a set $S$ and a vertex $u$ by $S\cup u$, and
the subtraction of $u$ from $S$ by $S\setminus u$.
For presentation simplicity, we omit the subscript $G$ from the
notations when the context is clear, and abbreviate $N_{G[S\cup u]}(u) =
N(u)\cap S$
as $N_S(u)$ and $\overline{N}_{G[S\cup u]}(u) = \overline{N}(u)\cap S$ as $\overline{N}_S(u)$.
For an arbitrary graph $g$, we denote its sets of vertices, edges and
non-edges by $V(g)$, $E(g)$ and $\overline{E}(g)$, respectively. 

\begin{definition}[$k$-Defective Clique]
	A graph $g$ is a $k$-defective clique if it misses at most $k$
	edges from being a clique, \ie, $|E(g)| \geq
	\frac{|V(g)|(|V(g)|-1)}{2} - k$ or
	equivalently, $|\overline{E}(g)| \leq k$.
\end{definition}

Obviously, if a subgraph $g$ of $G$ is a $k$-defective clique, then the
subgraph of $G$ induced by vertices $V(g)$ is also a $k$-defective
clique. Thus, {\em we refer to a $k$-defective clique simply by its set
of vertices}, and measure the size of a $k$-defective clique $S
\subseteq V$ by its number of vertices, \ie, $|S|$.
The property of $k$-defective clique is {\em hereditary}, \ie, any
subset of a $k$-defective clique is also a $k$-defective clique.
A $k$-defective clique $S$ of $G$ is a {\em maximal $k$-defective
clique} if every proper superset of $S$ in $G$ is not a $k$-defective
clique, and is a {\em maximum $k$-defective clique} if its size is the
largest among all $k$-defective cliques of $G$; denote the size of the
maximum $k$-defective clique of $G$ by $\omega_k(G)$.
Consider the graph in \cfig\ref{fig:graph}, both
$\{v_1,v_7,\ldots,v_8\}$ and $\{v_2,v_{11},\ldots,v_{14}\}$ are maximum
$2$-defective cliques with $\omega_2(G) = 5$; that is, the maximum
$k$-defective clique is not unique.

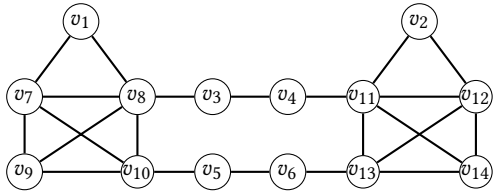
\begin{figure}[htb]
\centering
\begin{tikzpicture}[scale=0.5]
\node[draw, circle, inner sep=1.5pt] (1) at (1.5,2) {$v_1$};
\node[draw, circle, inner sep=1.5pt] (2) at (10.5,2) {$v_2$};
\node[draw, circle, inner sep=1.5pt] (3) at (5,0) {$v_3$};
\node[draw, circle, inner sep=1.5pt] (4) at (7,0) {$v_4$};
\node[draw, circle, inner sep=1.5pt] (5) at (5,-2) {$v_5$};
\node[draw, circle, inner sep=1.5pt] (6) at (7,-2) {$v_6$};
\node[draw, circle, inner sep=1.5pt] (7) at (0,0) {$v_7$};
\node[draw, circle, inner sep=1.5pt] (8) at (3,0) {$v_8$};
\node[draw, circle, inner sep=1.5pt] (9) at (0,-2) {$v_9$};
\node[draw, circle, inner sep=0pt] (10) at (3,-2) {$v_{10}$};
\node[draw, circle, inner sep=0pt] (11) at (9,0) {$v_{11}$};
\node[draw, circle, inner sep=0pt] (12) at (12,0) {$v_{12}$};
\node[draw, circle, inner sep=0pt] (13) at (9,-2) {$v_{13}$};
\node[draw, circle, inner sep=0pt] (14) at (12,-2) {$v_{14}$};
\path[draw,thick] (1) -- (7) -- (9) -- (10) -- (7) -- (8) -- (9);
\path[draw,thick] (1) -- (8) -- (10) -- (5) -- (6) -- (13) -- (14) --
(12) -- (2) -- (11) -- (14);
\path[draw,thick] (8) -- (3) -- (4) -- (11) -- (12) -- (13) -- (11);
\end{tikzpicture}
\caption{An example graph}
\label{fig:graph}
\end{figure}

\stitle{Problem Statement.} Given a graph $G = (V,E)$ and an integer $k
\geq 1$, we study the problem of maximum $k$-defective clique
computation, which aims to find the largest $k$-defective clique in $G$.

Frequently used notations are summarized in \ctab\ref{table:notations}.

\begin{table}[t]
\small
\caption{Frequently used notations}
\label{table:notations}
\begin{tabular}{cp{.7\columnwidth}} \hline
	Notation & Meaning \\ \hline
	$G = (V,E)$ & an unweighted, undirected and simple graph with vertex
	set $V$ and edge set $E$ \\
	$\omega_k(G)$ & the size of the maximum $k$-defective
	clique of $G$ \\
	$g = (V(g), E(g))$ & a subgraph of $G$ \\
	$S \subseteq V$ & a $k$-defective clique \\
	$(g,S)$ & a backtracking instance with $S\subseteq V(g)$ \\
	$N_S(u)$ & the set of $u$'s {\em neighbors} that are in $S$ \\
	$\overline{N}_S(u)$ & the set of $u$'s {\em non-neighbors} that are in $S$ \\
	$d_S(u)$ & the {\em number} of $u$'s neighbors that are in $S$ \\
	$E(S)$ & the set of {\em edges} induced by $S$ \\
	$\overline{E}(S)$ & the set of {\em non-edges} induced by $S$ \\
	$\gT, \gT'$ & search tree of backtracking algorithms \\
	$I, I', I_0, I_1,\ldots$ & nodes of the search tree $\gT$ or $\gT'$ \\
	$|I|$ & the size of $I$, \ie, $|V(I.g)\setminus I.S|$ \\
	$\ell_{\gT}(I)$ & number of leaf nodes in the subtree of $\gT$
	rooted at $I$ \\
	\hline
\end{tabular}
\end{table}

\section{The State-of-the-art Time Complexity}
\label{sec:existing}

In this section, we first review the state-of-the-art algorithm
\kdc~\cite{SIGMOD24:Chang} in \csec\ref{sec:kdc}, then briefly describe its
time complexity analysis in \csec\ref{sec:time_complexity_kdc}, and
finally discuss challenges of improving the time
complexity in \csec\ref{sec:challenge}.

\subsection{The Existing Algorithm \kdc}
\label{sec:kdc}

The problem of maximum $k$-defective clique computation is
NP-hard~\cite{STOC78:Yannakakis}. The existing exact algorithms compute
the maximum $k$-defective clique via {\em branch-and-bound} search (aka.
{\em backtracking}). Let $(g,S)$ denote a backtracking instance, where
$g$ is a (sub-)graph (of the input graph $G$) and $S \subseteq V(g)$ is a
$k$-defective clique in $g$. The goal of solving the instance $(g,S)$ is
to find the largest $k$-defective clique in the instance (\ie, in $g$
and containing $S$); thus, solving the instance $(G,\emptyset)$ finds
the maximum $k$-defective clique in $G$.
To solve an instance $(g,S)$, a backtracking algorithm selects a
branching vertex $b \in V(g)\setminus S$, and then recursively
solves two newly generated instances: one
includes $b$ into $S$, and the other removes $b$ from $g$. For
the base case that $S = V(g)$, $S$ is the maximum $k$-defective clique in
the instance.
For example, \cfig\ref{fig:full_tree} shows a snippet of the
backtracking search tree $\gT$, where each node corresponds to a
backtracking instance $(g,S)$. The two newly generated instances are
represented as the two children of the node, and the branching vertex
is illustrated on the edge; for the sake of simplicity,
\cfig\ref{fig:full_tree} only shows the branching vertices for the first
two levels.

\begin{figure}[htb]
\centering
\includegraphics[scale=.75]{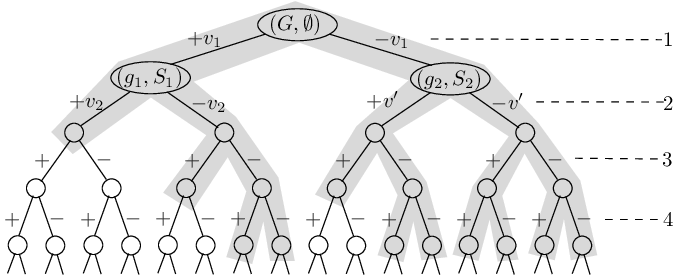}
\caption{A snippet of the (binary) search tree $\gT$ of a backtracking
algorithm}
\label{fig:full_tree}
\end{figure}

The state-of-the-art time complexity is achieved by
\kdc~\cite{SIGMOD24:Chang} which proposes a new
branching rule and two reduction rules to achieve the time complexity.
Specifically, \kdc proposes the non-fully-adjacent-first branching rule
\textbf{BR} preferring branching on a vertex that is not fully
adjacent to $S$, and the excess-removal reduction rule {\bf RR1} and
the high-degree reduction rule {\bf RR2}.
\begin{description}
	\item[BR~\cite{SIGMOD24:Chang}.] Given an instance $(g,S)$, the
		branching vertex is selected as the vertex of $V(g)\setminus S$
		that has at least one non-neighbor in $S$; if no such vertices
		exist, an arbitrary vertex of $V(g)\setminus S$ is chosen as
		the branching vertex.
	\item[RR1~\cite{SIGMOD24:Chang}.] Given an instance $(g,S)$, if a
		vertex $u \in
		V(g)\setminus S$ satisfies $|\overline{E}(S\cup u)| > k$, we
		can remove $u$ from $g$.
	\item[RR2~\cite{SIGMOD24:Chang}.] Given an instance $(g,S)$, if a
		vertex $u \in
		V(g)\setminus S$ satisfies $|\overline{E}(S\cup u)| \leq k$ and
		$d_g(u) \geq |V(g)|-2$, we can greedily add
		$u$ to $S$.
\end{description}

\subsection{Time Complexity Analysis of \kdc}
\label{sec:time_complexity_kdc}

The general idea of time complexity analysis is as follows.
As polynomial factors are usually ignored in the time complexity
analysis of exponential time algorithms, it is sufficient to bound the
number of leaf nodes of the search tree (in
\cfig\ref{fig:full_tree}) inductively in a
bottom-up fashion~\cite{Book2010:Fomin}. One way of bounding the
number of leaf nodes of the subtree rooted at the node
corresponding to instance $(g,S)$ is to order $V(g)\setminus S$ in such
a way that the longest prefix of the ordering that can be added to $S$
without violating the $k$-defective clique definition is short and bounded.
Specifically, let $(v_1,\ldots,v_l,v_{l+1},\ldots)$ be an ordering of
$V(g)\setminus S$ such that the longest prefix that can be added to $S$
without violating the $k$-defective clique definition is
$(v_1,\ldots,v_l)$; that is, $\{v_1,\ldots,v_l,v_{l+1}\} \cup S$ induces
more than $k$ non-edges.
Then, we only need to
generate $l+1$ new instances/branches, corresponding to the first $l+1$
prefixes, as shown in \cfig\ref{fig:multiway}: for the $i$-th (starting
from $0$) branch, we include $(v_1,\ldots,v_i)$ to $S$ and remove
$v_{i+1}$ from $g$. Denote the $i$-th branch by $(g_i,S_i)$. It holds
that 
\begin{itemize}
	\item $|V(g_i)\setminus S_i| \leq |V(g)\setminus S| - (i+1), \forall 0
		\leq i \leq l$. 
\end{itemize}
It can be shown by the techniques of~\cite{Book2010:Fomin}
that the number of leaf nodes of the search tree is at most $\gamma^n$
where $\gamma < 2$ is the largest real root of the equation
$x^{l_{\max}+2} -
2x^{l_{\max}+1}+1=0$ and $l_{\max}$ is the largest $l$ among all
non-leaf nodes. Thus, the smaller the value of $l_{\max}$, the
better the time complexity. 

\begin{figure}[htb]
\centering
	\includegraphics[scale=.8]{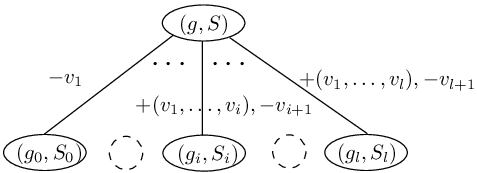}
\caption{General idea of time complexity analysis}
\label{fig:multiway}
\end{figure}

For \kdc, $l_{\max} = k+1$ and thus its time complexity is
$\bigo^*(\gamma_k^n)$ where $\gamma_k < 2$ is the largest real root of
the equation $x^{k+3} - 2x^{k+2}+1=0$~\cite{SIGMOD24:Chang}; here, the
$\bigo^*$ notation hides polynomial factors.
Specifically, \kdc orders $V(g)\setminus S$ by iteratively applying {\bf
BR}. That is, each time it appends to the ordering a vertex that has at
least one non-neighbor in either $S$ or the vertices already in the
ordering; if no such vertices exist, an arbitrary vertex of
$V(g)\setminus S$ is appended.
It is proved in~\cite{SIGMOD24:Chang} that after
exhaustively applying reduction rules \textbf{RR1} and \textbf{RR2}, the
resulting instance $(g,S)$ satisfies the condition:
\begin{itemize}
	\item $|\overline{E}(S\cup u)| \leq k$ and $d_g(u) < |V(g)|-2$,
		$\forall u \in V(g)\setminus S$.
\end{itemize}
\ie, \textbf{every vertex of $V(g)\setminus S$ has at least two
non-neighbors in $g$}. 
Then, the worst-case scenario (for time complexity) is that the
non-edges of $g[V(g)\setminus S]$ form a set of vertex-disjoint cycles;
a length-$(k+1)$ prefix of the ordering induces exactly $k$
non-edges, and a length-$(k+2)$ prefix induces more than $k$ non-edges.

One may notice that \cfig\ref{fig:full_tree} illustrates a {\em binary} search
tree while \cfig\ref{fig:multiway} shows a {\em multi-way} search tree.
Nevertheless, the
above techniques can be used to analyze \cfig\ref{fig:full_tree} since a
binary search tree can be (virtually) converted into an equivalent multi-way
search tree, which is the way the time complexity
of \kdc was analyzed in~\cite{SIGMOD24:Chang}. That is, we could collapse a
length-$l$ path in \cfig\ref{fig:full_tree} to make it have $l+1$
children. This will be more clear when we conduct our time complexity
analysis in \clem\ref{lemma:leafs_of_nonleaf_nodes}.

\subsection{Challenges of Improving the Time Complexity}
\label{sec:challenge}

As discussed in \csec\ref{sec:time_complexity_kdc}, the smaller the
value of $l_{\max}$, the better the time complexity.
\kdc~\cite{SIGMOD24:Chang} proves that $l_{\max} \leq k+1$ by making all
vertices of $V(g)\setminus S$ have at least two non-neighbors in $g$,
which is achieved by {\bf RR2}. In contrast, if all vertices of
$V(g)\setminus S$ have exactly one non-neighbor in $g$, then
$l_{\max}$ becomes $2k+1$ and the time complexity is
$\bigo^*(\gamma_{2k}^n)$, which is the case of \madec~\cite{COR21:Chen}.

It is natural to wonder whether the value of $l_{\max}$ can be 
reduced if we have techniques to make each vertex of $V(g)\setminus
S$ have more (than two) non-neighbors in $g$. Specifically, let's consider
the complement graph $\overline{g}$ of $g$: each edge of
$\overline{g}$ corresponds to a non-edge of $g$. The question is 
whether the {\bf BR} of \kdc can guarantee $l_{\max} < k+1$ when
$\overline{g}[V(g)\setminus S]$ has a minimum degree larger than two.
Unfortunately, the answer is negative.
It is shown in~\cite{JLMS63:Sachs} that for any $r \geq 2$ and $s \geq
3$, there exists a graph in which each vertex has exactly $r$ neighbors
and the shortest cycle has length exactly $s$; these graphs are called
$(r,s)$-graphs. Thus, when $\overline{g}[V(g)\setminus S]$ is an
$(r,s)$-graph for $s \geq k+2$, iteratively applying the branching rule
{\bf BR} may first identify vertices of the shortest cycle and it then
needs a prefix of length $k+2$ to cover $k+1$ edges of
$\overline{g}[V(g)\setminus S]$ (corresponding to $k+1$ non-edges of
$g$).
Alternatively, one may tempt to design a different
branching rule than {\bf BR} for finding a subset of $k+1$ or fewer
vertices $C \subseteq V(g)\setminus S$ such that $\overline{g}[C]$ has at
least $k+1$ edges. This most likely cannot be conducted efficiently, by
noting that it is NP-hard to
find a densest $k$-subgraph (\ie, a subgraph with exactly $k$ vertices
and the most number of edges) when $k$ is a part of the
input~\cite{Algorithmica01:Feige}.

\section{Our Algorithm with Improved Time Complexity}
\label{sec:approach}

Despite the challenges and negative results mentioned in
\csec\ref{sec:challenge}, we in this section show that we can reduce
both the base and the exponent of the time complexity.
In the following, we first present our
algorithm in \csec\ref{sec:algorithm}, then analyze its time complexity
in \csec\ref{sec:time_complexity}, and finally analyze the time
complexity again but using the degeneracy gap parameterization in
\csec\ref{sec:degen_gap}.

\subsection{Our Algorithm \kdct}
\label{sec:algorithm}

Our algorithm uses the same branching rule {\bf BR} and reduction
rules {\bf RR1} and {\bf RR2} as \kdc. But we will show in
\csec\ref{sec:time_complexity} that the base of the
time complexity is reduced by using different analysis techniques for
different nodes of the search tree.
Furthermore, we make use of the diameter-two property of large
$k$-defective cliques to reduce the exponent of the time complexity, by
observing that most real graphs have $\omega_k(G) \geq k+2$.

\begin{lemma}[Diameter-two Property of Large $k$-Defective
	Clique~\cite{COR21:Chen}]
\label{lemma:diameter_two}
For any $k$-defective clique, if it contains at least $k+2$ vertices,
its diameter is at most two (\ie, any two non-adjacent vertices
must have common neighbors in the defective clique).
\end{lemma}

Following \clem\ref{lemma:diameter_two}, if we know that $\omega_k(G)
\geq k+2$, then for a
backtracking instance $(g,S)$ with $S\neq \emptyset$, we can remove from
$g$ the vertices whose shortest distance (computed in $g$) to any
vertex of $S$ is greater than two. This could significantly reduce the
search space, as real graphs usually have a small average degree.
However, it is difficult to utilize the diameter-two property reliably,
since we do not know before-hand whether $\omega_k(G) \geq k+2$ 
or not and a $k$-defective clique of size
smaller than $k+2$ may have a diameter larger than two.
To resolve this, we propose to compute the maximum $k$-defective clique
in two stages, where Stage-I utilizes the diameter-two property for
pruning by assuming $\omega_k(G) \geq k+2$. If Stage-I fails (to find a
	$k$-defective clique of size at least
$k+2$), then we go to Stage-II searching the graph again without
utilizing the diameter-two property. This guarantees that the maximum
$k$-defective clique is found regardless of its size.

\begin{algorithm}[htb]
\small
\caption{$\kdct(G, k)$}
\label{alg:kdct}
\KwIn{A graph $G$ and an integer $k$}
\KwOut{A maximum $k$-defective clique in $G$}

\vspace{1.5pt}
\State{$C^* \la \emptyset$}
\State{Let $(v_1,\dots,v_n)$ be a degeneracy ordering of the vertices of
$G$}
\ForEach{$v_i \in V(G)$}{
	\State{$A \la N(v_i)\cap \{v_i,\ldots,v_n\}$}
	\State{Let $g_{v_i}$ be the subgraph of $G$ induced by $N[A]\cap
	\{v_{i},\ldots,v_n\}$}
	\State{$\bbsearch(g_{v_i},\{v_i\})$}
}
\lIf{$|C^*| < k+1$}{$\bbsearch(G, \emptyset)$}
\Return{$C^*$}

\vspace{2pt}
\nonl \textbf{Procedure} $\bbsearch(g, S)$ \\
\State{$(g', S') \la $ apply reduction rules {\bf RR1} and {\bf
RR2} to $(g,S)$}
\lIf{$g'$ is a $k$-defective clique}{update $C^*$ by $V(g')$ and
\textbf{return}}
\State{$b\la $ choose a branching vertex from $V(g')\setminus S'$ based
on {\bf BR}}
\StateCmt{$\bbsearch(g', S'\cup b)$}{\footnotesize Left branch includes $b$}
\StateCmt{$\bbsearch(g'\setminus b, S')$}{\footnotesize Right branch excludes
	$b$}
\end{algorithm}

The pseudocode of our algorithm \kdct is shown in
\calg\ref{alg:kdct} which takes a graph $G$ and an integer $k$ as input
and outputs a maximum $k$-defective clique of $G$; here, \kw{two} refers
to both ``two''-stage and diameter-``two''.
Let $C^*$ store the
currently found largest $k$-defective clique, which is initialized as
$\emptyset$ (Line~1). We first compute a degeneracy ordering of the
vertices of $G$ (Line~2). Without loss of generality, let
$(v_1,\ldots,v_n)$ be the degeneracy ordering, \ie, for each $1
\leq i \leq n$, $v_i$ is the vertex with the smallest degree in the
subgraph of $G$ induced by $\{v_i,\ldots,v_n\}$; the degeneracy ordering
can be computed in $\bigo(m)$ time by the
peeling algorithm~\cite{JACM83:Matula}. Then, for
each vertex $v_i \in V(G)$, we compute the largest diameter-two
$k$-defective clique in which the first vertex, according to the
degeneracy ordering, is $v_i$, by invoking the procedure $\bbsearch$
with input $(g_{v_i}, \{v_i\})$ (Lines~4--6); that is, the diameter-two
$k$-defective clique contains $v_i$ and is a subset of
$(v_i,v_{i+1},\ldots,v_n)$. Here, $g_{v_i}$ is the subgraph of $G$
induced by $v_i$ and its neighbors and two-hop neighbors that come later
than $v_i$ according to the degeneracy ordering. After that, we check
whether the currently found largest $k$-defective clique $C^*$ is of
size at least $k+1$: if $|C^*| \geq k+1$, then $C^*$ is guaranteed to be
a maximum $k$-defective clique of $G$; otherwise, $\omega_k(G) \leq k+1$
and the maximum $k$-defective clique of $G$ may have
a diameter larger than two. For the latter, we invoke \bbsearch
again with input $(G,\emptyset)$ (Line~7). Note that, we do not make use
of the diameter-two property for pruning within the procedure \bbsearch,
and thus ensure the correctness of our algorithm.

The pseudocode of the procedure $\bbsearch$ is shown at Lines~9--13 of
\calg\ref{alg:kdct}. Given an input $(g,S)$, we first apply reduction
rules {\bf RR1} and {\bf RR2} to reduce the instance $(g,S)$ to a
potentially smaller instance $(g',S')$ such that $V(g')\setminus S'
\subseteq V(g)\setminus S$
(Line~9). 
If $g'$ itself is a $k$-defective clique, then we update $C^*$ by
$V(g')$ and backtrack (Line~10).
Otherwise, we pick a branching vertex $b$ based on the branching rule
\textbf{BR} (Line~11), and generate two \bbsearch instances
and go into recursion (Lines~12--13).

\begin{figure}[htb]
\centering
\includegraphics[scale=.8]{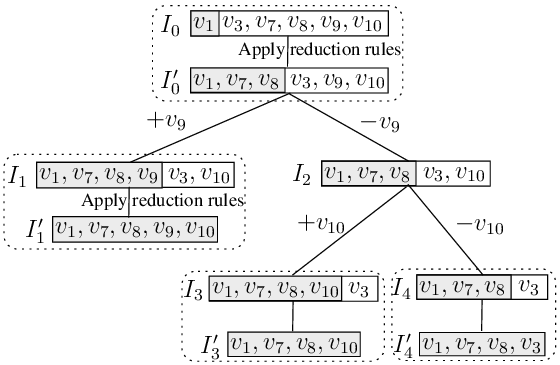}
\caption{Running example}
\label{fig:running_example}
\end{figure}

\begin{example}
Consider the graph $G$ in \cfig\ref{fig:graph}.
$(v_1,v_2,\ldots,v_{14})$ is a degeneracy ordering of $G$; for example,
$v_1$ has the smallest degree in $G$, and after removing $v_1$, $v_2$
then has the smallest degree, so on so forth. When processing $v_1$, we
only need to consider its neighbors $\{v_7,v_8\}$ and two-hop neighbors
$\{v_3,v_9,v_{10}\}$; that is, we only need to consider the subgraph
$g_{v_1}$ that is induced by $\{v_1,v_3,v_7,\ldots,v_{10}\}$. The
backtracking search tree for $\bbsearch(g_{v_1}, \{v_1\})$ is shown in
\cfig\ref{fig:running_example}, where vertices of $S$ for each instance
$(g,S)$ are in the shaded area. For the root node $I_0$, applying reduction
rule {\bf RR2} adds $v_7$ and $v_8$ into $S$ since each of them has at most
one non-neighbor in the subgraph; this results into the instance $I'_0$.
Suppose the branching rule {\bf BR} selects $v_9$ for $I'_0$ as $v_9$ is not
adjacent to $v_1$. Then, two new instances $I_1$ and $I_2$ are generated
as the two children of $I'_0$. For $I_1$, applying the reduction rule
{\bf RR1} removes $v_3$ from the graph and then applying {\bf RR2} adds
$v_{10}$ into $S$; consequently we reach a leaf node. Similarly,
$v_{10}$ is selected as the branching vertex for $I_2$, which then
generates two leaf nodes $I'_3$ and $I'_4$.
\end{example}

\subsection{Time Complexity Analysis of \kdct}
\label{sec:time_complexity}

To analyze the time complexity of \kdct, we use the same terminologies
and notations as~\cite{SIGMOD24:Chang} and consider the search tree
$\gT$ of (recursively) invoking \bbsearch at either Line~6 or Line~7 of
\calg\ref{alg:kdct}, as shown in \cfig\ref{fig:full_tree}.
To avoid confusion, nodes of the search tree are referred to by {\em
nodes}, and vertices of a graph by {\em vertices}.
Nodes of $\gT$ are denoted by $I, I', I_0, I_1,\ldots$, and the
graph $g$ and the partial solution $S$ of the \bbsearch instance to
which $I$ corresponds are, respectively, denoted by $I.g$ and $I.S$.
Note that {\em $I.g$ and $I.S$ denote the ones obtained after applying
the reduction rules} at Lines~9--10 of \calg\ref{alg:kdct}, where
Line~10 is regarded as applying the reduction rule that if $g'$ is
a $k$-defective clique, then all vertices of $V(g')\setminus S'$ are
added to $S'$; that is, $I_0,I_1,I_3,I_4$ in
\cfig\ref{fig:running_example} are instances before applying the
reduction rules and could be discarded from the search tree.
The size of $I$ is measured by the number of vertices of $I.g$
that are not in $I.S$, \ie, $|I| = |V(I.g)|-|I.S| =
|V(I.g)\setminus I.S|$. It is worth mentioning that 
\begin{itemize}
	\item $|I'| \leq |I|-1$ whenever $I'$ is a child of $I$, \eg, the
		branching vertex $b$ of $I$ is in $V(I.g)\setminus I.S$ but not
		in $V(I'.g)\setminus I'.S$.
	\item $|I|=0$ whenever $I$ is a leaf node
\end{itemize}

For a non-leaf node $I_0$, we consider the path
$(I_0,I_1,\ldots,I_q)$ that starts from $I_0$, always visits the left
child, and stops at $I_q$ once $q \geq 1$ and
$|V(g_q)| < |V(g_{q-1})|$; see \cfig\ref{fig:left_deep} for an example.
Note that, the path is well defined since $I_0$ is a non-leaf node and
thus $I_0.g$ is not a $k$-defective clique.
For presentation simplicity, we abbreviate $I_i.g$ and $I_i.S$ as $g_i$
and $S_i$, respectively, and denote the branching vertex selected for
$I_i$ by $b_i$.
Let $u$ be an arbitrary vertex of $V(g_{q-1})\setminus V(g_{q})$. Then,
$(b_0,\ldots,b_{q-1},u,\ldots)$ is an ordering of $V(g_0)\setminus S_0$
such that adding $(b_0,\ldots,b_{q-1},u)$ to $S_0$
violates the $k$-defective clique definition.
In the lemma below, we prove that $q \leq k$ if $I_0$
has at least one branching vertex being added.

\begin{figure}[htb]
\centering
\includegraphics[scale=0.7]{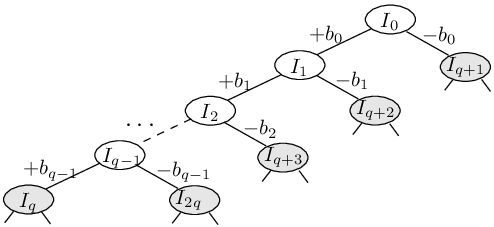}
\caption{A left-deep walk in the search tree $\gT$ starting from $I_0$;
	$I_q$ is the first
node such that $q \geq 1$ and $|V(g_q)| < |V(g_{q-1})|$}
\label{fig:left_deep}
\end{figure}

\begin{lemma}
\label{lemma:bound_q}
If the non-leaf node $I_0$ has at least one branching vertex being added,
then $q \leq k$.
\end{lemma}

\begin{proof}
Let $I_x$ be the last node, on the path $(I_0, I_1, \ldots,I_x,\ldots,
I_{q-1})$, satisfying the condition that all vertices of
$V(g_x)\setminus S_x$ are adjacent to all vertices of $S_x$, \ie,
the branching vertex $b_x$ selected for $I_x$ has no non-neighbors in
$S_x$. If such an $I_x$ does not exist, then we have
$|\overline{E}(S_{i+1})| \geq |\overline{E}(S_{i})|+1$ for all $0 \leq i
< q$ (because the branching vertex added to $S_{i+1}$ must bring at
least one non-edge to $\overline{E}(S_{i+1})$), and consequently $q \leq
|\overline{E}(S_q)| \leq k$.
In the following, we assume that such an $I_x$ exists, and prove that
$|\overline{E}(S_x)| > x$ by considering two cases.
Note that, each of the branching vertices $b_i \in
\{b_0,\dots,b_{x-1}\}$ that are added to $S_x$
must have at least two non-neighbors in $g_i$ (because of {\bf RR2})
and all these non-neighbors are in $S_x$
(according to the definitions of $I_x$ and $I_q$).
\begin{itemize}
	\item \textbf{Case-I}: $|\overline{E}(S_0)|\neq 0$. Then,
		the number of unique non-edges associated with
		$\{b_0,\ldots,b_{x-1}\}$ is at least $x$, and
		$|\overline{E}(S_x)| \geq |\overline{E}(S_0)| + x > x$.
	\item \textbf{Case-II}: $|\overline{E}(S_0)| = 0$. Let $b_{-1}$ be
		the branching vertex added to $S_0$ from its parent
		(which exists according to the lemma statement). Then, $b_{-1}$
		has at least two non-neighbors in $g_0$; note that these
		non-neighbors will not be removed from $g_0$ by \textbf{RR1}
		since $|\overline{E}(S_0)|=0$ and all vertices of
		$V(g_0)\setminus S_0$ are fully adjacent to vertices of
		$S_0\setminus b_{-1}$. Thus, the
		number of non-edges between $\{b_{-1}, b_0,\ldots,b_{x-1}\}$ is
		at least $x+1$, and hence $|\overline{E}(S_x)| > x$.
\end{itemize}
Then, according to the definition of $I_x$ and our branching rule, for
each $i$ with $x+1\leq i < q$, the branching vertex $b_i$
selected for $I_i$ has at least one non-neighbor in
$S_i$, and consequently, 
\[
	|\overline{E}(S_q)| \geq |\overline{E}(S_x)| + (q-x-1) > q-1
\]
Thus, the lemma follows from the fact that $|\overline{E}(S_q)| \leq k$.
\end{proof}

Let $\ell_{\gT}(I)$ denote the number of leaf nodes in the
subtree of $\gT$ rooted at $I$. We prove in the lemma below that
$\ell_{\gT}(I) \leq \beta_k^{|I|}$ when at least one branching vertex
has been added to $I$. Note that $\beta_k = \gamma_{k-1}$.

\begin{lemma}
	\label{lemma:leafs_of_nonleaf_nodes}
	For any node $I$ of $\gT$ that has at least one branching vertex being
	added, it holds that $\ell_{\gT}(I) \leq \beta_k^{|I|}$ where $1 <
	\beta_k< 2$ is the largest real root of the equation
	$x^{k+2}-2x^{k+1}+1=0$.
\end{lemma}

\begin{proof}
We prove the lemma by induction.
For the base case that $I$ is a leaf node, it is trivial that
$\ell_{\gT}(I) =
1 \leq \beta_k^{|I|}$ since $\beta_k > 1$ and $|I| = 0$.
For a non-leaf node $I_0$, let's consider the path $(I_0,
I_1,\ldots,I_{q'})$ that starts from $I_0$, always visits the left child in
the search tree $\gT$, and stops at $I_{q'}$ once $q' \geq 1$ and
$|I_{q'}| \leq |I_{q'-1}|-2$. Note that $(I_0,\ldots,I_{q'})$ is a
prefix of $(I_0,\ldots,I_q)$ since $I_q$
satisfies the condition that $q \geq 1$ and $|I_q| \leq |I_{q-1}|-2$.
It is trivial that 
\[
	\ell_{\gT}(I_0) = \ell_{\gT}(I_{q'}) + \ell_{\gT}(I_{q'+1}) + \ell_{\gT}(I_{q'+2}) +
	\cdots + \ell_{\gT}(I_{2q'})
\]
where $I_{q'+1},I_{q'+2},\ldots,I_{2q'}$ are the right child of
$I_0,I_1,\ldots,I_{q'-1}$, respectively, as illustrated in
\cfig\ref{fig:left_deep} (by replacing $q$ with $q'$); {\em this is
equivalent to converting a binary search tree to a multi-way search tree
by collapsing the path $(I_0,\ldots,I_{q'-1})$ into a super-node that has
$I_{q'},I_{q'+1},\ldots,I_{2q'}$ as its children}.
To bound $\ell_{\gT}(I_0)$, we need to bound $q'$ and $|I_i|$ for $q'
\leq i \leq 2q'$.
Following from \clem\ref{lemma:bound_q}, we have
\begin{description}
	\item[Fact 1.] $q' \leq q \leq k$.
\end{description}
Also, according to the definition of the path, it holds that 
\begin{equation}
	\forall i \in [1,q'-1],\quad
	S_i = S_{i-1}\cup b_{i-1},\quad V(g_i) = V(g_{i-1})
	\label{eq:path}
\end{equation}
That is, the reduction rules at Line~9 of
\calg\ref{alg:kdct} have no effect on $I_i$ for $1 \leq i < q'$.
Then, the following two facts hold.
\begin{description}
	\item[Fact 2.] $\forall i \in [q'+1, 2q']$,
		$|I_i| \leq |I_{i-q'-1}|-1 \leq |I_0| + q'-i$.
	\item[Fact 3.] $|I_{q'}| \leq |I_{q'-1}|-2 \leq |I_0|-q'-1$.
\end{description}
Based on Facts 1, 2 and 3, we have 
\begin{align*}
	\ell_{\gT}(I_0) & = \ell_{\gT}(I_{q'+1}) + \ell_{\gT}(I_{q'+2}) +
	\cdots + \ell_{\gT}(I_{2q'}) +
	\ell_{\gT}(I_{q'}) \nonumber \\
	& \leq \beta_k^{|I_{q'+1}|} + \beta_k^{|I_{q'+2}|} + \cdots +
	\beta_k^{|I_{2q'}|} + \beta_k^{|I_{q'}|} \nonumber \\
	& \leq \beta_k^{|I_0|-1} + \beta_k^{|I_0|-2} + \cdots +
	\beta_k^{|I_0|-q'} + \beta_k^{|I_0|-q'-1} \nonumber \\
	& \leq \beta_k^{|I_0|-1} + \beta_k^{|I_0|-2} + \cdots +
	\beta_k^{|I_0|-k} + \beta_k^{|I_0|-k-1}
	\label{eq:time_complexity_proof}
\end{align*}
where $\beta_k^{|I_0|-1} + \beta_k^{|I_0|-2} + \cdots + \beta_k^{|I_0|-k}
+ \beta_k^{|I_0|-k-1} \leq \beta_k^{|I_0|}$ if $\beta_k$ is
no smaller than the largest real root of the equation $x^{k+1} -
x^{k} -\cdots - x - 1=0$ which is equivalent to the equation
$x^{k+2}-2x^{k+1}+1=0$~\cite{Book2010:Fomin}. The first few solutions to
the equation are $\beta_1= 1.619$, $\beta_2= 1.840$, $\beta_3= 1.928$,
$\beta_4=1.966$, and $\beta_5=1.984$.
\end{proof}

\begin{figure}[htb]
\centering
\includegraphics[scale=0.7]{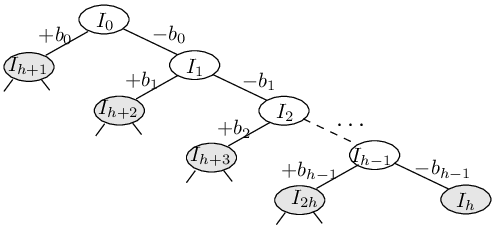}
\caption{A right-deep walk in the search tree $\gT$ starting from $I_0$
and stopping at a leaf node $I_h$}
\label{fig:search_tree_right}
\end{figure}

In \clem\ref{lemma:leafs_of_nonleaf_nodes}, we cannot bound
$\ell_{\gT}(I)$ by $\beta_k^{|I|}$ if no branching vertices have been
added to $I$. Nevertheless, we prove in the lemma below that
$\ell_{\gT}(I) < 2\cdot \beta_k^{|I|}$ holds for every node $I$ of
$\gT$, by using a non-inductive proving technique.

\begin{lemma}
	\label{lemma:leafs}
	For any node $I$ of $\gT$, it holds that $\ell_{\gT}(I) <
	2\cdot \beta_k^{|I|}$.
\end{lemma}

\begin{proof}
If $I$ is a leaf node, the lemma is trivial. For a non-leaf node
$I_0$, let's consider the path $(I_0,I_1,\ldots,I_h)$ that starts from
$I_0$, always visits the right child in the search tree $\gT$, and
stops at a leaf node $I_h$; see \cfig\ref{fig:search_tree_right}
for an example. Then, it holds that
\[
	\forall i \in [h+1, 2h],
		|I_i| \leq |I_{i-h-1}|-1 \leq |I_0| + h-i
\]
and $|I_h| = 0$. Moreover, for each $i \in [h+1,2h]$, $I_i$ has at least
one branching vertex (\eg, $b_{i-h-1}$) being added, and thus according
to \clem\ref{lemma:leafs_of_nonleaf_nodes}, it satisfies $\ell_{\gT}(I_i) \leq
\beta_k^{|I_i|}$. Consequently,
\begin{align*}
	\ell_{\gT}(I_0) & = \ell_{\gT}(I_{h+1}) + \ell_{\gT}(I_{h+2}) +
	\cdots + \ell_{\gT}(I_{2h}) +
	\ell_{\gT}(I_h) \nonumber \\
	& \leq \beta_k^{|I_{h+1}|} + \beta_k^{|I_{h+2}|} + \cdots +
	\beta_k^{|I_{2h}|} + \beta_k^{|I_h|} \nonumber \\
	& \leq \beta_k^{|I_0|-1} + \beta_k^{|I_0|-2} + \cdots +
	\beta_k^{|I_0|-h} + 1 \nonumber \\
	& \textstyle \leq \frac{\beta_k^{|I_0|}-1}{\beta_k-1} + 1 < 2
	\cdot
	\beta_k^{|I_0|}
\end{align*}
The last inequality follows from the fact that $\beta_k> 1.5, \forall k
\geq 1$.
\end{proof}

Following \clem\ref{lemma:leafs}, Line~7 of \calg\ref{alg:kdct} runs in
$\bigo^*(\beta_k^n)$ time.
What remains is to bound $|I|$ for Line~6 of \calg\ref{alg:kdct}.
Let $f$ denote the vertex obtained at Line~3 of \calg\ref{alg:kdct}, 
$g$ be the subgraph extracted at Line~5, and $t$ be the number of
$f$'s non-neighbors in $g$. It holds that $d_g(f) = |A| \leq
\alpha$ since a degeneracy ordering is used for extracting $g$, $t \leq \alpha (\Delta-1)$, and $|V(g)| \leq d_g(f)+t+1 \leq
\alpha\Delta+1$; here, $\alpha$ is the degeneracy of $G$ and $\Delta$ is
the maximum degree of $G$. We prove the time complexity of \kdct in the
theorem below.

\begin{theorem}
Given a graph $G$ and an integer $k$, \kdct runs in $\bigo(n\times
(\alpha\Delta)^2 \times \beta_k^{\alpha\Delta})$ time when $\omega_k(G)
\geq k+2$, and in $\bigo(m\times \beta_k^n)$ time otherwise.
\end{theorem}

\begin{proof}
Each invocation of \bbsearch at Line~6 of \calg\ref{alg:kdct} takes
$\bigo( (\alpha\Delta)^2 \times \beta_k^{\alpha\Delta})$ time, since the
root node $I$ of the search tree satisfies $|I| \leq \alpha\Delta$ and
each node of the search tree (\ie, Lines~9--11 of \calg\ref{alg:kdct})
takes $\bigo(|E(g)|) = \bigo( (\alpha\Delta)^2)$ time. When
$\omega_k(G)\geq k+2$, the condition of Line~7
is not satisfied and thus \calg\ref{alg:kdct} runs in $\bigo(n\times
(\alpha\Delta)^2 \times \beta_k^{\alpha\Delta})$ time. Otherwise, Line~7
takes $\bigo(m\times \beta_k^n)$ time and \calg\ref{alg:kdct} runs in
$\bigo(n\times (\alpha\Delta)^2 \times \beta_k^{\alpha\Delta} + m\times
\beta_k^n)$; for simplicity, we abbreviate this time complexity as
$\bigo(m\times \beta_k^n)$ since it is usually the dominating term.
\end{proof}

\subsubsection{Further Reduce the Exponent}
\label{sec:reduce_exponent}

In the following, we show that when $\omega_k(G) \geq k+2$, the exponent
of the time complexity can be further reduced by a more refined analysis.
Let $I_0=(g_0,S_0)$ be the root of $\gT$ for Line~6 of
\calg\ref{alg:kdct}; recall that, $f \in S_0$,
$d_{g_0}(f) \leq d_g(f) \leq \alpha$ and $|I_0| \leq \alpha+t$, where
$f$ is the vertex obtained at Line~3 of \calg\ref{alg:kdct} and $t$ is
number of $f$'s non-neighbors in the subgraph $g$ extracted for $f$.
Let's consider the subtree of $\gT$ formed by
starting a depth-first-search from $I_0$ and backtracking once the
path from $I_0$ to it has either $t$ total edges or $k$ positive edges (\ie,
labeled as ``+''); see
the shaded subtree in \cfig\ref{fig:full_tree} for an illustration of 
$t=4$ and $k=2$. Let $\gL$ be the set of leaf nodes of this subtree. Then,
the number of leaf nodes of $\gT$ satisfies
\[
	\ell_{\gT}(I_0) \leq \sum_{I \in \gL} \ell_{\gT}(I)
\]
We bound $|\gL|$ and $\ell_{\gT}(I)$ for $I \in \gL$ in the following
two lemmas.

\begin{lemma}
	$|\gL| = \bigo(t^k)$ where $t$ is the number of $f$'s non-neighbors
	in $g$.
\end{lemma}

\begin{proof}
To bound $|\gL|$, we observe that the search tree $\gT$ is a full binary
tree with each node having a positive edge to its left child and a
negative edge to its right child. Thus, we can label every edge by its
level in the tree (see \cfig\ref{fig:full_tree}), and for each node $I
\in \gL$, we associate with it a set of numbers corresponding to the
levels of the positive edges from $I_0$ to $I$. Then, it is easy to see
that each node of $\gL$ is associated with a distinct subset, of size at
most $k$, of $\{1,2,\ldots,t\}$. Consequently, $|\gL| \leq
\sum_{i=0}^{k} {t\choose i} = \bigo(t^{k})$~\cite{Book12:Mohri}.
\end{proof}

\begin{lemma}
	\label{lemma:bound_instance_size}
	For all $I \in \gL$, it holds that $|I| \leq \alpha$ and $|V(I.g)|
	\leq \alpha+k+1$.
\end{lemma}

\begin{proof}
Let's consider the path $(I_0,I_1,\ldots,I_{p-1},I_p=I)$ from $I_0$ to
$I$.
If there are $k$ positive edges on the path, then all vertices of
$V(g_p)\setminus S_p$ must be adjacent to $f$, and thus $|I_p| \leq
\alpha$ and $|V(g_p)| \leq \alpha+k+1$. The latter holds since all of
$f$'s non-neighbors in $g$ are in $S$ and thus of quantity at most $k$.
The former can be shown by
contradiction. Suppose $V(g_p)\setminus
S_p$ contains a non-neighbor of $f$, then adding each of the $k$
branching vertices (on the positive edges of the path) must bring at least one
non-edge to $S_p$, due to the branching rule {\bf BR}; then \textbf{RR1} will
remove all non-neighbors of $f$ from $V(g_p)\setminus S_p$,
contradiction.

Otherwise, there are at most $k-1$ positive edges on the path and $p=t$.
Then, $|I_p| \leq |I_0|-t \leq \alpha$. Also, there are at least
$t-k+1$ negative edges on the path and thus $|V(g_p)| \leq |V(g_0)| -
(t-k+1) \leq t+\alpha+1-(t-k+1) = \alpha+k$.
\end{proof}

Consequently, the number of leaf nodes of $\gT$ is $\bigo(
(\alpha\Delta)^k \beta_k^\alpha)$ and the time complexity of \kdct
follows.

\begin{theorem}
\label{theorem:time_complexity}
Given a graph $G$ and an integer $k$, \kdct runs in $\bigo(n\times
(\alpha \Delta)^{k+2} \times \beta_k^{\alpha})$ time when $\omega_k(G)
\geq k+2$.
\end{theorem}

\subsubsection{Compared with \kdc} 

Our algorithm \kdct improves the time complexity of
\kdc~\cite{SIGMOD24:Chang} from two aspects.
Firstly, we improve the base of the exponential time complexity from
$\gamma_k$ to $\beta_k = \gamma_{k-1}$.
This is achieved by using different analysis
techniques for the nodes that already have branching vertices being
added (\ie, \clem\ref{lemma:leafs_of_nonleaf_nodes}) and for those that
do not (\ie, \clem\ref{lemma:leafs}).
Without this separation, we can only bound the length $q$
of the path $(I_0,\ldots,I_q)$ by $k+1$ instead of $k$ that is proved in
\clem\ref{lemma:bound_q}.
Also note that, \clem\ref{lemma:leafs} is non-inductive and has a
coefficient $2$ in the bound; if we use induction in the proof of
\clem\ref{lemma:leafs}, then the coefficient will become bigger and
bigger and become exponential when going up the tree.
Secondly, we improve the exponent of the time complexity from $n$ to
$\alpha$ when $\omega_k(G) \geq k+2$; note that, most real graphs
have $\omega_k(G)\geq k+2$. This is achieved by our two-stage
algorithm that utilizes the diameter-two property for pruning in
Stage-I, as well as our refined analysis in \csec\ref{sec:reduce_exponent}.

\subsection{Parameterize by the Degeneracy Gap}
\label{sec:degen_gap}

In this subsection, we prove that \kdct, with slight modification, runs
in $\bigo^*( (\alpha\Delta)^{k+2} (k+1)^{\alpha+k+1-\omega_k(G)})$ time
by using the degeneracy gap $\alpha+k+1-\omega_k(G)$ parameterization;
this is better than $\bigo^*( (\alpha\Delta)^{k+2} \beta_k^\alpha)$ when
$\omega_k(G)$ is close to $\alpha+k+1$, the degeneracy-based upper bound
of $\omega_k(G)$.
Specifically, let's consider the problem of testing whether $G$ has a
$k$-defective clique of size $\tau$.
To do so, we truncate the
search tree $\gT$ by cutting the entire subtree rooted at node $I$ if
$|V(I.g)| \leq \tau$; that is, we terminate $\bbsearch$ once $|V(g')|
\leq \tau$ after Line~10 of \calg\ref{alg:kdct}. Let $\gT'$ be the
truncated version of $\gT$. We first bound the number of leaf nodes of
$\gT'$ in the following two lemmas, in a similar fashion as
Lemmas~\ref{lemma:leafs_of_nonleaf_nodes} and \ref{lemma:leafs}.

\begin{lemma}
	\label{lemma:leafs_of_nonleaf_nodes2}
	For any node $I$ of $\gT'$ that has at least one branching vertex
	being added, it holds that $\ell_{\gT'}(I) \leq (k+1)^{|V(I.g)|-\tau}$.
\end{lemma}

\begin{proof}
We prove the lemma by induction. For the base case that $I$ is a leaf
node, we have $|V(I.g)| = \tau$ and thus $\ell_{\gT'}(I) = (k+1)^0 =
(k+1)^{|V(I.g)|-\tau}$.
For a non-leaf node $I_0$ that has at least one branching vertex being
added, let's consider the path
$(I_0,I_1,\ldots,I_q)$ that starts from $I_0$, always visits the left
child in the search tree $\gT'$, and stops at $I_q$ once $q \geq 1$ and 
$|V(g_q)| < |V(g_{q-1})|$; this is the same one as studied at the
beginning of \csec\ref{sec:time_complexity}, see \cfig\ref{fig:left_deep}.
According to \clem\ref{lemma:bound_q}, we have $q \leq k$. It is
trivial that
\[
	\forall i \in [q,2q],\quad |V(g_i)| < |V(g_0)|
\]
where $I_{q+1},I_{q+2},\dots,I_{2q}$ are the right child of
$I_0,I_1,\dots,I_{q-1}$, respectively, as illustrated in
\cfig\ref{fig:left_deep}. 
Thus,
\begin{align*}
	\ell_{\gT'}(I_0) & = \ell_{\gT'}(I_q) + \ell_{\gT'}(I_{q+1}) +
	\ell_{\gT'}(I_{q+2}) + \dots + \ell_{\gT'}(I_{2q}) \\
	& \leq (q+1) \cdot (k+1)^{|V(g_0)|-1-\tau} \leq
	(k+1)^{|V(g_0)|-\tau} 
\end{align*}
\end{proof}

\begin{lemma}
	\label{lemma:leafs_degen_gap}
	For any node $I$ of $\gT'$, it holds that
	$\ell_{\gT'}(I) < 2\cdot (k+1)^{|V(I.g)|-\tau}$.
\end{lemma}

\begin{proof}
If $I$ is a leaf node, the lemma is trivial. For a non-leaf node
$I_0$, let's consider the path $(I_0,I_1,\ldots,I_h)$ that starts from
$I_0$, always visits the right child in the search tree $\gT'$, and
stops at a leaf node $I_h$; see \cfig\ref{fig:search_tree_right}
for an example. Then, it holds that
\[
	\forall i \in [h+1, 2h],
	|V(g_i)| \leq |V(g_{i-h-1})| \leq |V(g_0)| + h+1-i
\]
and $|V(g_h)| \leq |V(g_0)| - h \leq \tau$. Moreover, for each $i \in
[h+1,2h]$, $I_i$ has at least
one branching vertex (\eg, $b_{i-h-1}$) being added, and thus according
to \clem\ref{lemma:leafs_of_nonleaf_nodes2}, it satisfies $\ell_{\gT'}(I_i) \leq
(k+1)^{|V(g_i)|-\tau} \leq (k+1)^{|V(g_0)|+h+1-i-\tau}$. Consequently,
\begin{align*}
	\ell_{\gT'}(I_0) & = \ell_{\gT'}(I_{h+1}) + \ell_{\gT'}(I_{h+2}) +
\cdots + \ell_{\gT'}(I_{2h}) +
\ell_{\gT'}(I_h) \nonumber \\
	& \leq (k+1)^{|V(g_0)|-\tau} + (k+1)^{|V(g_0)|-1-\tau} + \cdots +
	(k+1)^1 + 1 \\
	& \textstyle \leq \frac{(k+1)^{|V(g_0)|+1-\tau}-1}{k+1-1} < 2\cdot
	(k+1)^{|V(g_0)|-\tau}
\end{align*}
The last inequality follows from the fact that $k \geq 1$.
\end{proof}

Then, the following time complexity 
can be proved in a similar way to
\cthm\ref{theorem:time_complexity} but using
\clem\ref{lemma:leafs_degen_gap} and
\clem\ref{lemma:bound_instance_size}.

\begin{lemma}
Testing whether $G$ has a $k$-defective clique of size $\tau$ for $\tau
\geq k+2$ takes $\bigo(n\times (\alpha\Delta)^{k+2} \times
(k+1)^{\alpha+k+1-\omega_k(G)})$ time.
\end{lemma}

Finally, we can find the maximum $k$-defective clique by iteratively
testing whether $G$ has a $k$-defective clique of size $\tau$ for $\tau
= \{\alpha+k+1,\alpha+k,\ldots\}$. This will find the maximum $k$-defective
clique and terminate after testing $\tau=\omega_k(G)$. 
Consequently, the following time complexity follows.

\begin{theorem}
The maximum $k$-defective clique in $G$ can be found in
$\bigo( (\alpha+k+2-\omega_k(G))\times n \times (\alpha
\Delta)^{k+2}\times (k+1)^{\alpha+k+1-\omega_k(G)})$ time when
$\omega_k(G) \geq k+2$.
\end{theorem}

\section{A New Reduction Rule}
\label{sec:rr}

In this section, we propose a new reduction rule based on the
degree-sequence-based upper bound {\bf UB} that is proposed
in~\cite{AAAI22:Gao}, to further improve the practical performance of
\kdct.
\begin{description}
	\item[UB~\cite{AAAI22:Gao}.] Given an instance $(g,S)$, let $v_1,v_2,\ldots$ be an
		ordering of $V(g)\setminus S$ in non-decreasing order regarding
		their numbers of non-neighbors in $S$, \ie
		$|\overline{N}_S(\cdot)|$. The maximum $k$-defective
		clique in the instance $(g,S)$ is of size at most $|S|$ plus the
		largest $i$ such that $\sum_{j=1}^i |\overline{N}_S(v_j)| \leq k -
		|\overline{E}(S)|$.
\end{description}
Note that, different tie-breaking techniques for ordering the vertices
lead to the same upper bound. Thus, an arbitrary tie-breaking technique
can be used in {\bf UB}.

Let $\lb$ be the size of the currently found best solution. If an upper
bound computed by {\bf UB} for an instance $(g,S)$ is no larger than
$\lb$, then we can prune the instance. However, this way of first
generating an instance and then try to prune it based on a computed
upper bound is inefficient.
To improve efficiency, we propose to 
remove $u \in V(g)\setminus S$ from $g$ if an upper bound of
$(g, S\cup u)$ is no larger than $\lb$.
Note that, rather than computing the exact upper bound for $(g,S\cup
u)$, we only need to test whether the upper bound is larger than $\lb$
or not. The latter can be conducted more efficiently and without generating
$(g,S\cup u)$; moreover, computation can be shared
between the testing for different vertices of $V(g)\setminus S$.

Let $v_1,v_2,\ldots$ be an ordering of $V(g)\setminus (S\cup u)$ in
non-decreasing order regarding $|\overline{N}_S(\cdot)|$, and $C$ be the
set of vertices that have the same number of non-neighbors in $S$ as
$v_{\lb-|S|}$, \ie, $C = \{v_i \in V(g)\setminus (S\cup u) \mid
|\overline{N}_{S}(v_i)| = |\overline{N}_S(v_{\lb-|S|})|\}$. Let $C_1$
and $C_2$ be a partitioning of $C$ according to their positions in the
ordering, \ie, $C_1 = C \cap \{v_1,\ldots,v_{\lb-|S|}\}$ and $C_2 =
C\setminus C_1$.
Note that, both $C_1$ and $C_2$ contain consecutive vertices in the
ordering, and $C_2$ could be empty.
A visualization of the ordering and vertex sets is
shown below, where $S\cup u$ and $\lb-|S|$ are denoted by $R$ and $r$
for brevity.
\begin{equation*}
	\underbrace{S,u}_{\text{$R$}},
	\underbrace{v_1,\ldots,v_{r-|C_1|}}_{\text{$D$}}, \overbrace{
		\underbrace{v_{r-|C_1|+1},\ldots,v_{r}}_{\text{$C_1$}},
	\underbrace{v_{r+1},\ldots
v_{r+|C_2|}}_{\text{$C_2$}}}^\text{$C = \{v_i \in V(g)\setminus R
\mid |\overline{N}_S(v_i)| = |\overline{N}_S(v_{r})|\}$}, \ldots
\end{equation*}
We prove in the lemma below that the upper bound computed by {\bf UB}
for the instance $(g,S\cup u)$ is at most $\lb$ if and only if
\begin{equation}
	\label{eq:rr}
	|\overline{E}(S)| + \sum_{j = 1}^{r}
	|\overline{N}_S(v_j)| + |\overline{N}_{S\cup D}(u)| + \max\{
	|\overline{N}_{C_1}(u)|-|N_{C_2}(u)|, 0\} > k
\end{equation}

\begin{lemma}
\label{lemma:RR3}
The upper bound computed by {\bf UB} for the instance $(g,S\cup u)$ is
at most $\lb$ if and only if \cequ{\ref{eq:rr}} is satisfied.
\end{lemma}

\begin{proof}
Let's firstly simply the left-hand side (LHS) of \cequ{\ref{eq:rr}}. Note
that, $|\overline{N}_{S\cup D}(u)| = |\overline{N}_S(u)| +
|\overline{N}_D(u)|$, $|\overline{E}(S)| + |\overline{N}_S(u)| =
|\overline{E}(R)|$, and $\sum_{j=1}^{r-|C_1|} |\overline{N}_S(v_j)| +
|\overline{N}_D(u)| = \sum_{j=1}^{r-|C_1|} |\overline{N}_R(v_j)|$. 
Let $a = |\overline{N}_{C_1}(u)|$ and $b = |N_{C_2}(u)|$.
Then, \cequ{\ref{eq:rr}} is equivalent to
\begin{equation}
	\label{eq:eq1}
	|\overline{E}(R)| + \sum_{j=1}^{r-|C_1|} |\overline{N}_R(v_j)| +
	\sum_{v_j \in C_1} |\overline{N}_S(v_j)| + \max\{a-b, 0\} > k
\end{equation}
Let $w_1,w_2,\ldots$ be an ordering of $V(g)\setminus R$ in
non-decreasing order regarding $|\overline{N}_R(\cdot)|$. Then, the
upper bound computed by {\bf UB} for the instance $(g,R)$
is at most $\lb$ if and only if 
\begin{equation}
	\label{eq:eq2}
	|\overline{E}(R)| + \sum_{j=1}^{r} |\overline{N}_R(w_j)| > k
\end{equation}
For every $v \in V(g)\setminus R$, if $(u,v) \in E$, then
$|\overline{N}_R(v)| = |\overline{N}_S(v)|$; otherwise,
$|\overline{N}_R(v)| = |\overline{N}_S(v)|+1$.
According to the definition of $C$, for every $v \in D$, it holds that
$|\overline{N}_S(v)| < |\overline{N}_S(v_{r-|C_1|+1})|$ and thus
$|\overline{N}_R(v)| \leq |\overline{N}_S(v_{r-|C_1|+1})| \leq
|\overline{N}_R(w_{r-|C_1|+1})|$.
Consequently, we can assume $\{w_1,\ldots,w_{r-|C_1|}\} = D$; note
however that, the ordering may be different.
Similarly, we can assume $\{w_{r-|C_1|+1},\ldots,w_{r}\} \subseteq
C_1\cup C_2$. Then, we have
\begin{align*}
	& \text{LHS of \cequ{\ref{eq:eq1}}} - \text{LHS of \cequ{\ref{eq:eq2}}} \\
	= & \sum_{j = r-|C_1|+1}^r |\overline{N}_S(v_j)| + \max\{a-b,
	0\} - \sum_{j = r-|C_1|+1}^r |\overline{N}_R(w_j)| \\
	= & |C_1|\cdot |\overline{N}_S(v_{r})| + \max \{a-b, 0\} -
	\sum_{j= r-|C_1|+1}^r |\overline{N}_R(w_j)|
\end{align*}
There are exactly $|C_1|-a+b$ vertices of $C_1\cup C_2$
satisfying $|\overline{N}_R(\cdot)| = |\overline{N}_S(v_{r})|$. 
Thus, if $a \leq b$, then there are $|C_1|-a+b \geq |C_1|$ vertices of
$C_1\cup C_2$ with $|\overline{N}_R(\cdot)| = |\overline{N}_S(v_r)|$ and
$\sum_{j = r-|C_1|+1}^r |\overline{N}_R(w_j)| = |C_1|\cdot
|\overline{N}_S(v_{r})|$; otherwise, $\sum_{j =r-|C_1|+1}^r
|\overline{N}_R(w_j)| = |C_1|\cdot |\overline{N}_S(v_{r})| + a-b$.
Therefore, the lemma holds.
\end{proof}

Based on the above discussions, we propose the following
degree-sequence-based reduction rule {\bf RR3}.
\begin{description}
	\item[RR3 (degree-sequence-based reduction rule).] Given an instance
		$(g,S)$ with $|S| < \lb < |V(g)|$ and a vertex $u \in V(g)\setminus
		S$, we remove $u$ from $g$ if \cequ{\ref{eq:rr}} is satisfied.
\end{description}

\begin{algorithm}[htb]
\small
\caption{$\kw{ApplyRR3}(g,S,\lb)$}
\label{alg:rr3}
\State{Obtain $|\overline{E}(S)|$ and $|\overline{N}_S(v)|$ for each $v
\in V(g)\setminus S$}
\State{Let $u_1,u_2,\ldots$ be an ordering of $V(g)\setminus S$ in
non-decreasing order regarding $|\overline{N}_S(\cdot)|$}
\State{$X \la \emptyset$}
\For{$i \la 1 \textbf{ to } |V(g)\setminus S|$}{
	\lIf{$|X| + |V(g)\setminus S| - i < \lb-|S|$}{\textbf{return} $(g[S],S)$}
	\State{Let $v_1,v_2,\ldots$ be the vertices of $X\cup
	\{u_{i+1},u_{i+2},\ldots\}$ in non-decreasing order regarding
$|\overline{N}_S(\cdot)|$}
	\State{Obtain $|\overline{N}_D(u_i)|$, $|\overline{N}_{C_1}(u_i)|$ and
	$|N_{C_2}(u_i)|$}
	\lIf{\cequ{\ref{eq:rr}} is not satisfied}{Append $u_i$ to the end of
	$X$}
}
\Return{$(g[X], S)$}
\end{algorithm}

Given an instance $(g,S)$, our pseudocode of efficiently applying {\bf
RR3} for all vertices of $V(g)\setminus S$ is shown in
\calg\ref{alg:rr3}, which returns a reduced instance at either Line~5 or
Line~9. We first obtain $|\overline{E}(S)|$ and $|\overline{N}_S(v)|$
for each $v \in V(g)\setminus S$ (Line~1), and then order vertices of
$V(g)\setminus S$ in non-decreasing order regarding
$|\overline{N}_S(\cdot)|$ (Line~2). Let $u_1,u_2,\ldots$ be the ordered
vertices. We then process the vertices of $V(g)\setminus S$ one-by-one
according to the sorted order (Line~4). When processing $u_i$, the
vertices that we need to consider are the vertices that have not been
processed yet (\ie, $\{u_{i+1},u_{i+2},\ldots\}$) and the subset of
$\{u_1,u_2,\ldots,u_{i-1}\}$ that passed (\ie, not pruned) by the
reduction rule; denote the latter subset by $X$. This means that
$\{u_1,\ldots,u_{i-1}\}\setminus X$ have already been removed from $g$ by the
reduction rule. If the number of remaining vertices (\ie, $|X| +
|V(g)\setminus S| - i$) is less than $r=\lb-|S|$, then we remove all
vertices of $V(g)\setminus S$ from $g$ by returning $(g[S],S)$ (Line~5).
Otherwise, let $v_1,v_2,\ldots,v_r,\ldots$ be the vertices of $X\cup
\{u_{i+1},u_{i+2},\ldots\}$ in non-decreasing order regarding
$|\overline{N}_S(\cdot)|$ (Line~6). We obtain $|\overline{N}_D(u_i)|$,
$|\overline{N}_{C_1}(u_i)|$ and $|N_{C_2}(u_i)|$ (Line~7), and 
remove $u_i$ from $g$ if \cequ{\ref{eq:rr}} is satisfied; otherwise, $u_i$ is not pruned and is
appended to the end of $X$ (Line~8).
Finally, Line~9 returns the reduced instance $(g[X],S)$.

\begin{lemma}
	\calg\ref{alg:rr3} runs in $\bigo(|E(g)|+k)$ time. 
\end{lemma}

\begin{proof}
Firstly,
Lines~1--2 run in $\bigo(|E(g)|+k)$ time, where the sorting is
conducted by counting sort. Secondly, Line~6 does not do anything; it
is just a syntax sugar for relabeling the vertices.
Thirdly, Line~7 can be conducted in $\bigo(|N_g(u_i)|)$ time since
$|\overline{N}_D(u_i)| = |D|-|N_D(u_i)|$ and $|\overline{N}_{C_1}(u_i)|
= |C_1| - |N_{C_1}(u_i)|$; note that, as each of $D, C_1, C_2$ spans at
most two arrays (\ie, $X$ and $\{u_{i+1},\ldots\}$), we can easily get
its size and boundary.
Lastly, Line~8 can be checked in constant time by noting that
$\sum_{j=1}^r |\overline{N}_S(v_j)|$ can be obtained in constant time
after storing the suffix sums of $|\overline{N}_S(u_{1})|,
|\overline{N}_S(u_{2})|, \ldots$, $|\overline{N}_S(u_{i+1})|, |\overline{N}_S(u_{i+2})|, \ldots$. 
\end{proof}

\begin{figure}[htb]
\centering
\begin{tikzpicture}[scale=0.5]
\node[draw, circle, inner sep=1.5pt] (1) at (0,0) {$s_1$};
\node[draw, circle, inner sep=1.5pt] (2) at (0,-1.5) {$s_2$};
\node[draw, circle, inner sep=1.5pt] (3) at (0,-3) {$s_3$};
\node[draw, circle, inner sep=1.5pt] (4) at (4,1.5) {$u_1$};
\node[draw, circle, inner sep=1.5pt] (5) at (7,1) {$u_2$};
\node[draw, circle, inner sep=1.5pt] (6) at (9,0) {$u_3$};
\node[draw, circle, inner sep=1.5pt] (7) at (9,-1.5) {$u_4$};
\node[draw, circle, inner sep=1.5pt] (8) at (9,-3) {$u_5$};
\path[draw,thick] (2) -- (3);
\path[draw,thick] (1) -- (4) -- (2) -- (5) -- (3) -- (4) -- (5) -- (1);
\path[draw,thick] (2) -- (6) -- (3);
\path[draw,thick] (1) -- (7) -- (3);
\path[draw,thick] (1) -- (8) -- (2);
\path[draw,thick] (4) -- (7);
\end{tikzpicture}
\caption{Example instance for applying reduction rule {\bf RR3}}
\label{fig:rr}
\end{figure}
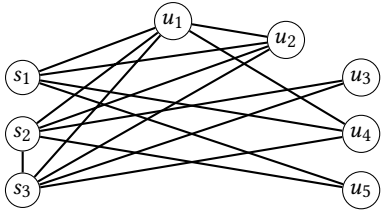

\begin{example}
Consider the instance $(g,S)$ in \cfig\ref{fig:rr} for $k=3$ and
$\lb=5$, where $g$ is the entire graph and $S = \{s_1,s_2,s_3\}$; thus
$r=2$. The values of $|\overline{N}_S(\cdot)|$ for the vertices of
$V(g)\setminus
S = \{u_1,\ldots,u_5\}$ are $\{u_1: 0, u_2: 0, u_3: 1, u_4: 1, u_5:1
\}$. As $|\overline{E}(S)| = 2$, the upper bound of $(g,S)$ computed by
{\bf UB} is $6$; thus, the instance is not pruned. 

Let's apply {\bf RR3}
for $u_1$. As $X = \emptyset$, we have $v_i = u_{i+1}$ for $1
\leq i \leq 4$. Then $D = \{u_2\}, C_1 = \{u_3\}, C_2 = \{u_4,u_5\}$,
$\sum_{j=1}^r |\overline{N}_S(v_j)| = 1$, $|\overline{N}_{S\cup D}(u_1)|
= 0$, $|\overline{N}_{C_1}(u_1)| = 1$ and $|N_{C_2}(u_1)| = 1$.
Thus, \cequ{\ref{eq:rr}} is not satisfied; $u_1$ is not
pruned and is appended to $X$.

Now let's apply {\bf RR3} for $u_2$. As $X = \{u_1\}$, we have $v_1 =
u_1$ and $v_i = u_{i+1}$ for $2 \leq i \leq 4$. Then $D = \{u_1\}, C_1 =
\{u_3\}, C_2 = \{u_4,u_5\}$, $\sum_{j=1}^r |\overline{N}_S(v_j)| = 1$,
$|\overline{N}_{S\cup D}(u_2)| = 0$, $|\overline{N}_{C_1}(u_2)| = 1$ and
$|N_{C_2}(u_2)| = 0$. Consequently, \cequ{\ref{eq:rr}} is satisfied and
$u_2$ is removed from $g$. It can be verified that $u_3,u_4,u_5$ will
all subsequently be removed.
\end{example}

\stitle{Effectiveness of {\bf RR3}.}
\kdc~\cite{SIGMOD24:Chang} also proposed a reduction rule based on {\bf
UB}; let's denote it as {\bf RR3'}. We remark that our {\bf
RR3} is more effective (\ie, prunes more vertices) than {\bf RR3'}
since the latter ignores the non-edges between $u$ and
$V(g)\setminus (S\cup u)$ that are considered by {\bf RR3}; specifically, {\bf
RR3'} removes $u$ from $g$ if $|\overline{E}(S)| + \sum_{j=1}^{r}
|\overline{N}_S(v_j)| + |\overline{N}_S(u)| > k$. 
More generally, the effectiveness of our {\bf RR3} is characterized by the
lemma below.

\begin{lemma}
{\bf RR3} is more effective than any other reduction rule that is
designed based on an upper bound of $(g,S\cup u)$ that ignores
all the non-edges between vertices of $V(g)\setminus (S\cup u)$. 
\end{lemma}

\begin{proof}
Firstly, we have proved in \clem\ref{lemma:RR3} that applying {\bf RR3}
is equivalent to computing {\bf UB} for $(g,S\cup u)$. Secondly, it can
be shown that {\bf UB} computes the tightest upper bound for $(g,S\cup
u)$ among all upper bounds that ignore all the non-edges between
vertices of $V(g)\setminus (S\cup u)$. Thus, the lemma holds.
\end{proof}

In particular, the second-order reduction rule proposed
in~\cite{SIGMOD24:Chang} is designed based on an upper bound of
$(g,S\cup u)$ that does not consider the non-edges between vertices of
$V(g)\setminus (S\cup u)$. Thus, {\bf RR3} is more effective than the
second-order reduction rule of~\cite{SIGMOD24:Chang}.

\begin{figure*}[t]
\centering
\begin{minipage}{.48\textwidth}
\subfigure[$k=1$]{
	\includegraphics[width=.5\columnwidth]{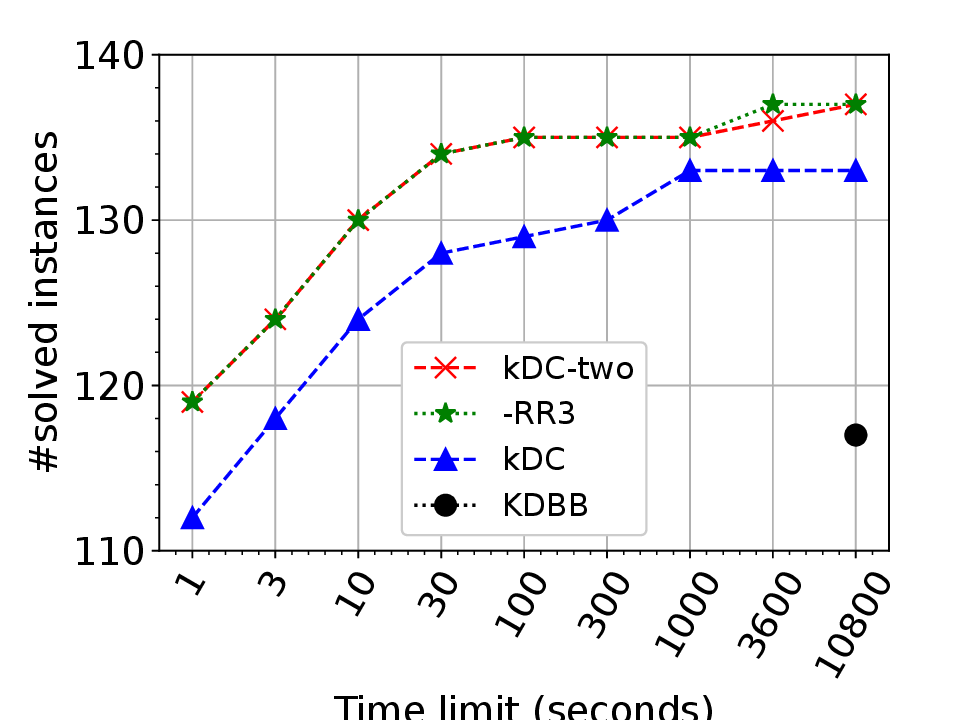}
}%
\subfigure[$k=3$]{
	\includegraphics[width=.5\columnwidth]{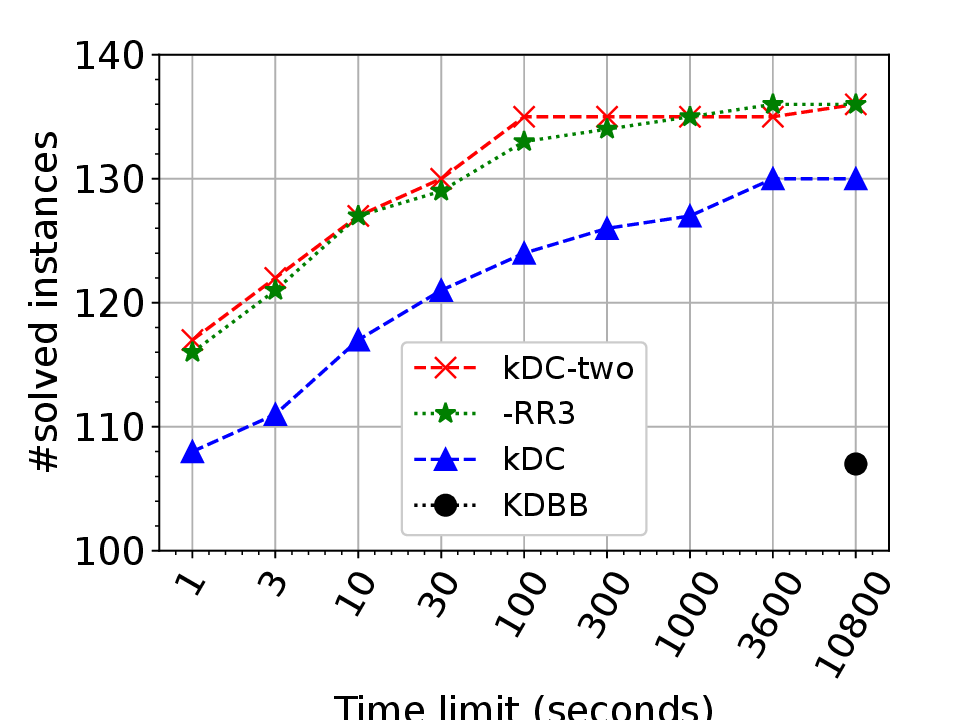}
}
\subfigure[$k=5$]{
	\includegraphics[width=.5\columnwidth]{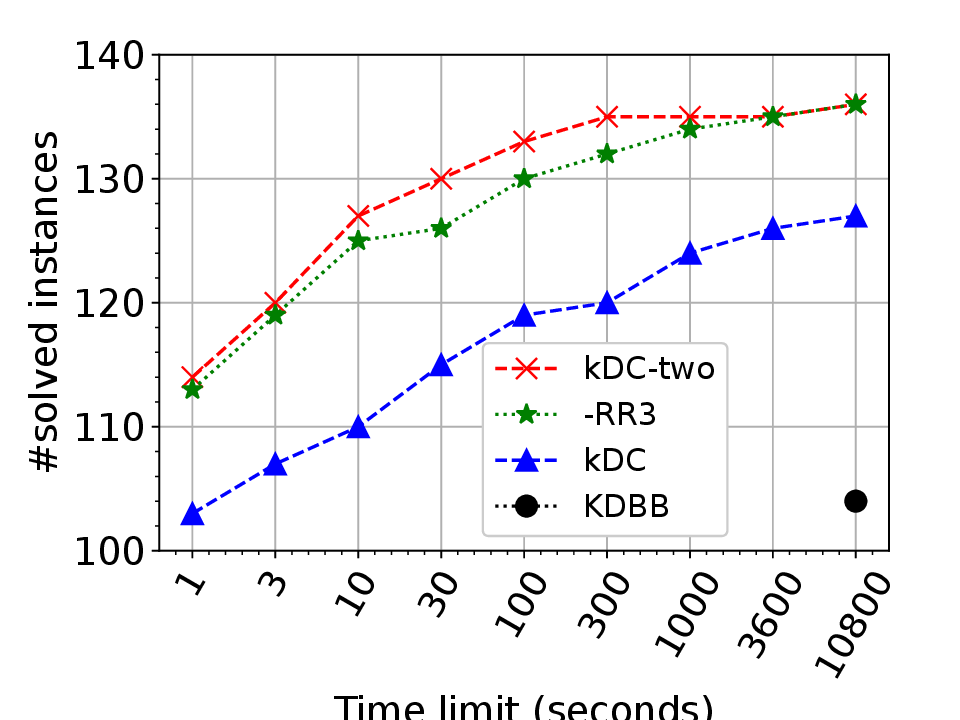}
}%
\subfigure[$k=10$]{
	\includegraphics[width=.5\columnwidth]{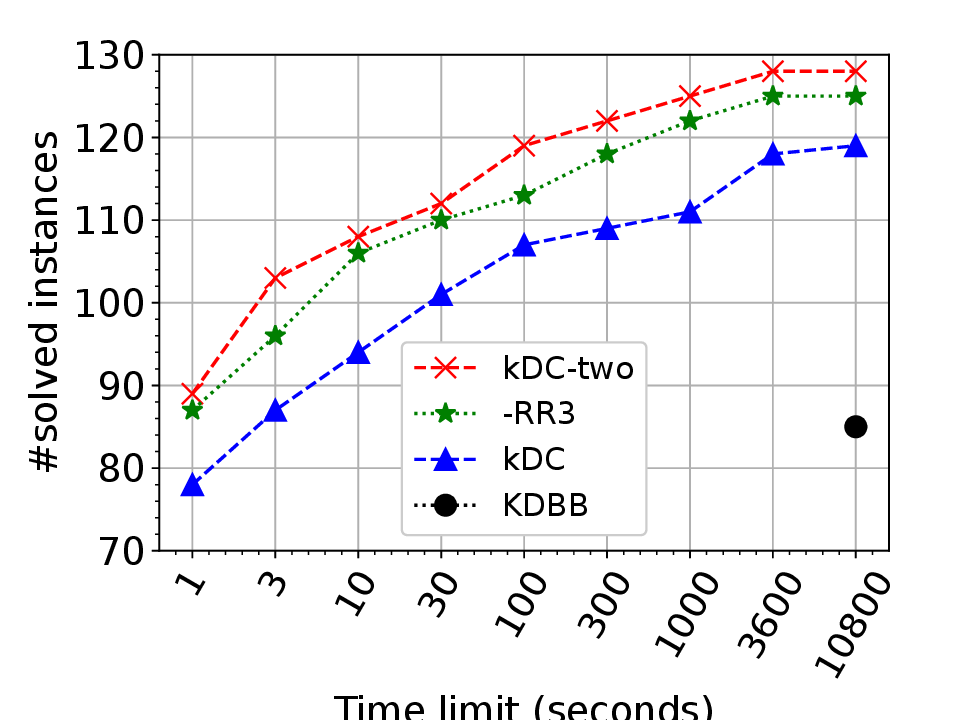}
}
\subfigure[$k=15$]{
	\includegraphics[width=.5\columnwidth]{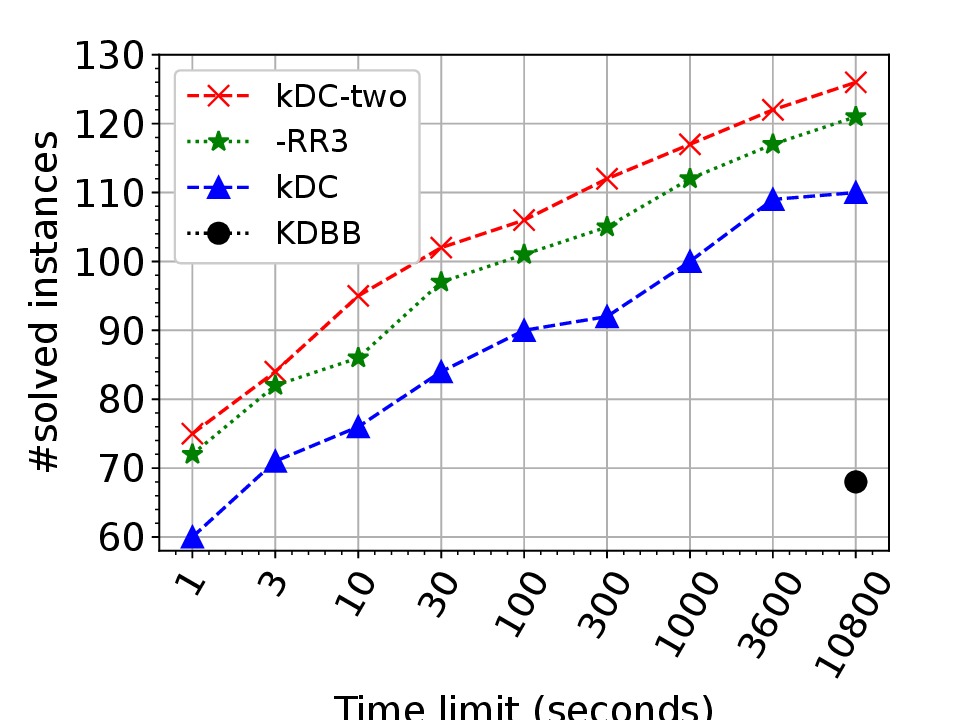}
}%
\subfigure[$k=20$]{
	\includegraphics[width=.5\columnwidth]{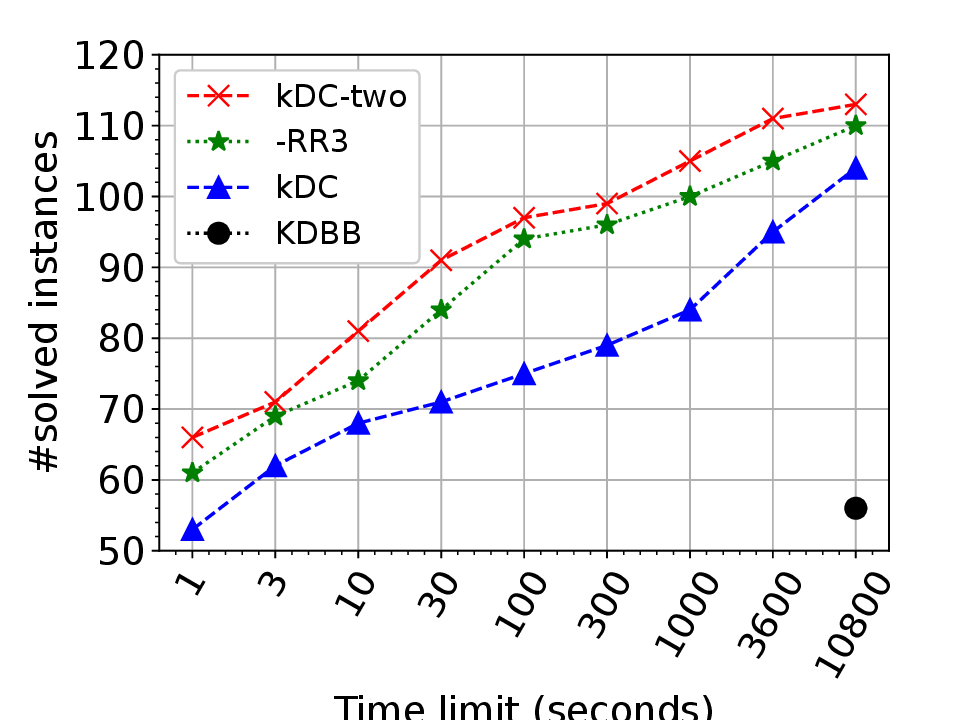}
}
\vspace{-2pt}
\caption{Number of solved instances by varying time limit for real-world graphs (best viewed in color)}
\label{fig:realworld}
\end{minipage}
\hspace*{5pt}
\begin{minipage}{.48\textwidth}
\subfigure[$k=1$]{
	\includegraphics[width=.5\columnwidth]{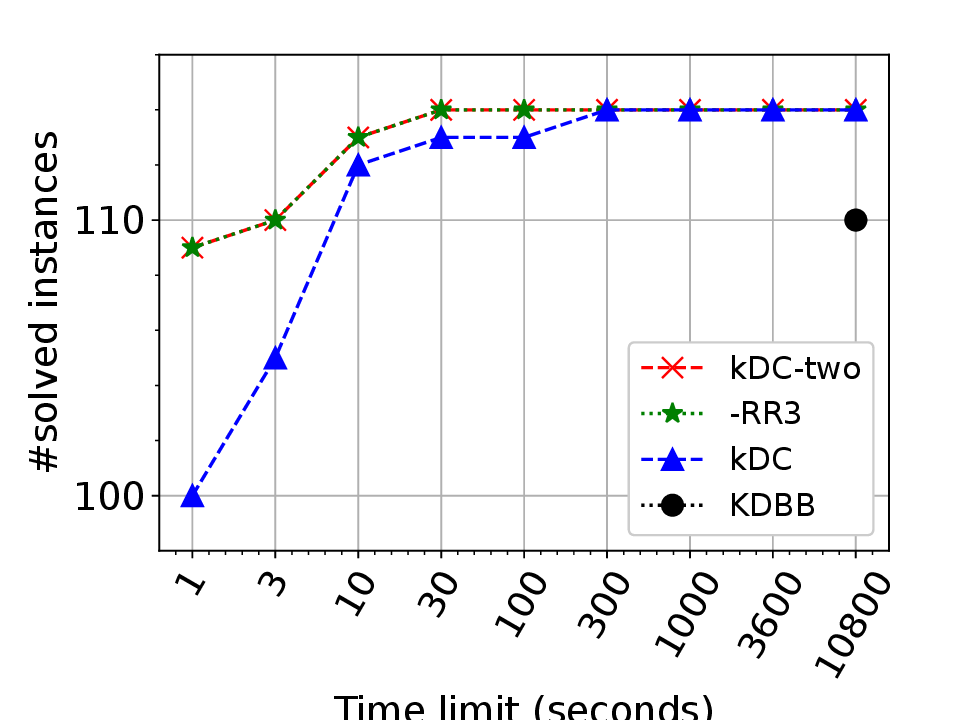}
}%
\subfigure[$k=3$]{
	\includegraphics[width=.5\columnwidth]{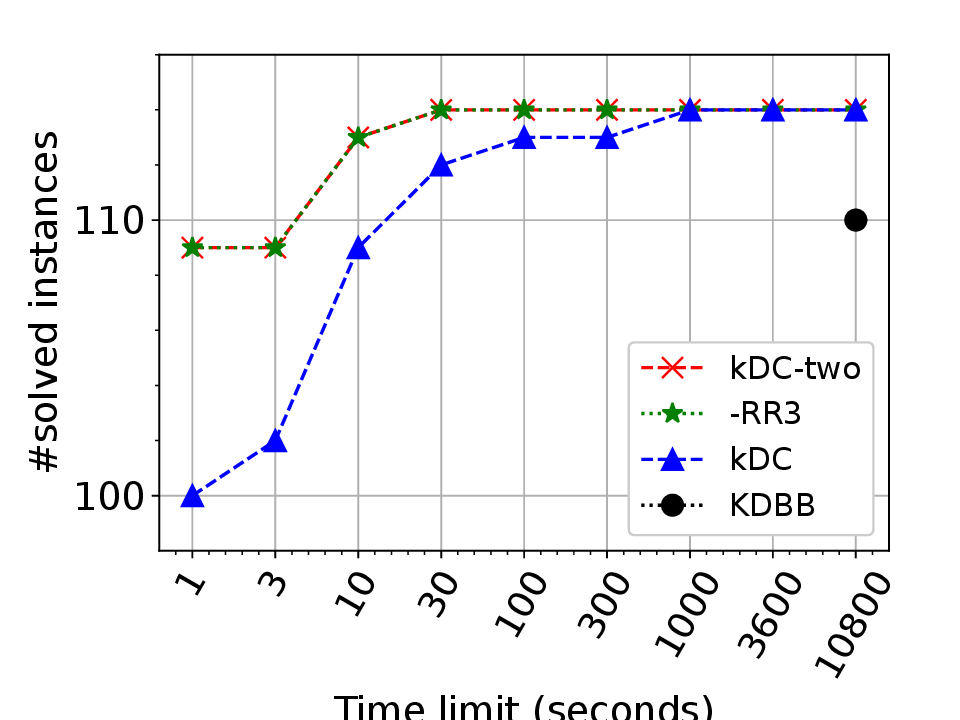}
}
\subfigure[$k=5$]{
	\includegraphics[width=.5\columnwidth]{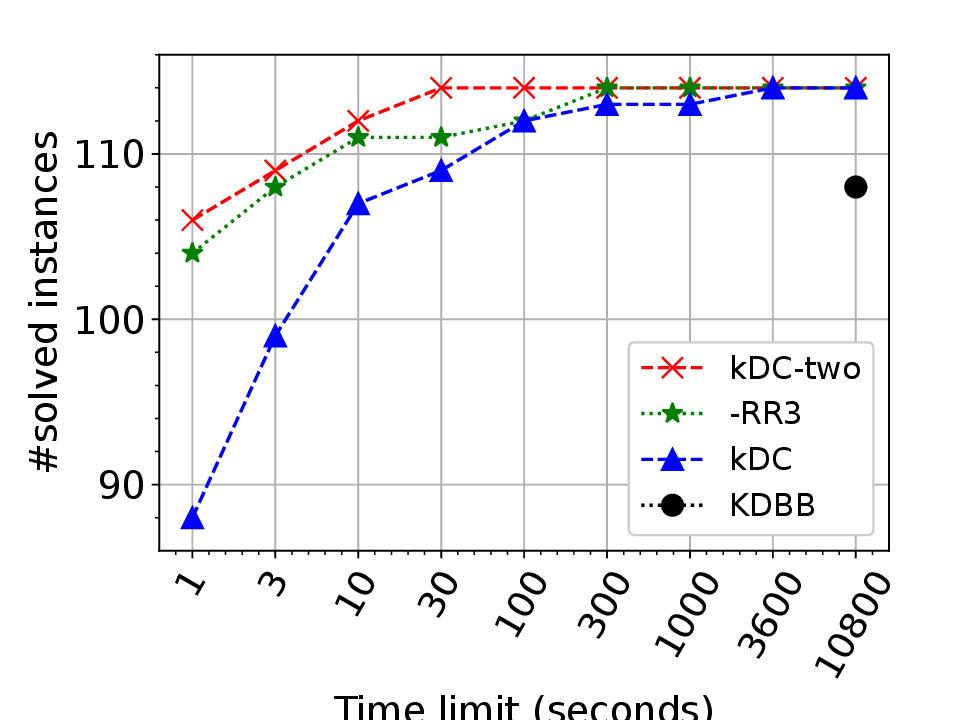}
}%
\subfigure[$k=10$]{
	\includegraphics[width=.5\columnwidth]{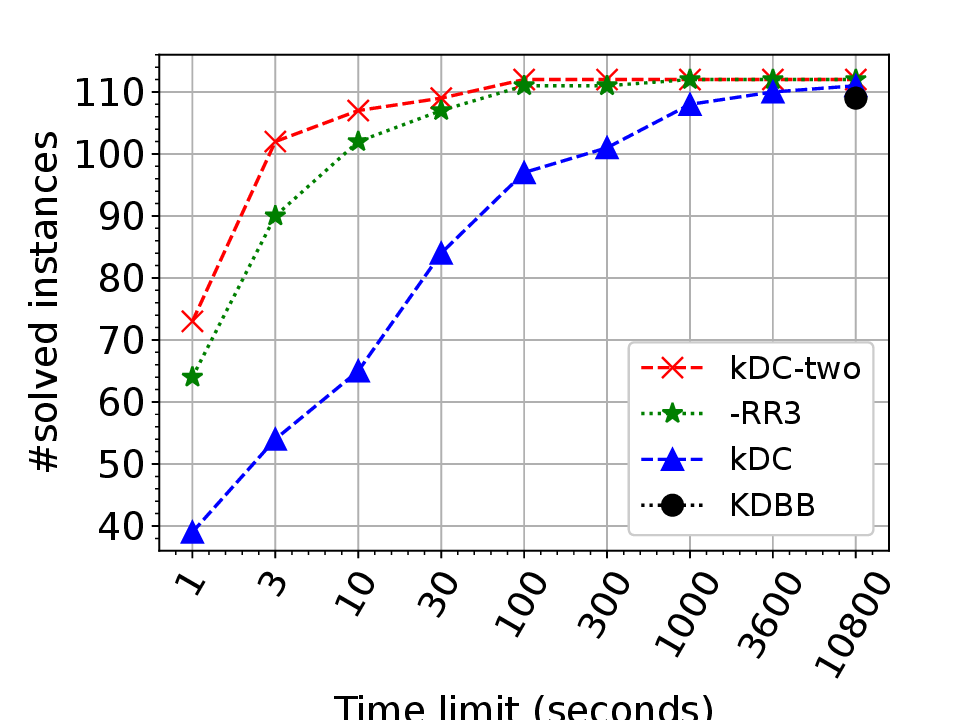}
}
\subfigure[$k=15$]{
	\includegraphics[width=.5\columnwidth]{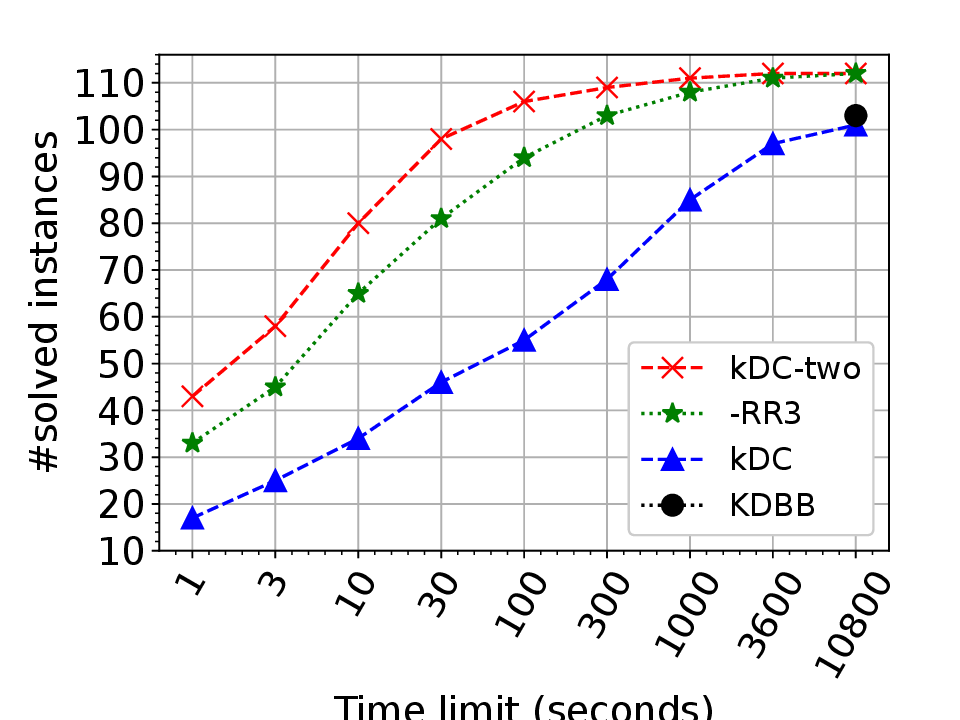}
}%
\subfigure[$k=20$]{
	\includegraphics[width=.5\columnwidth]{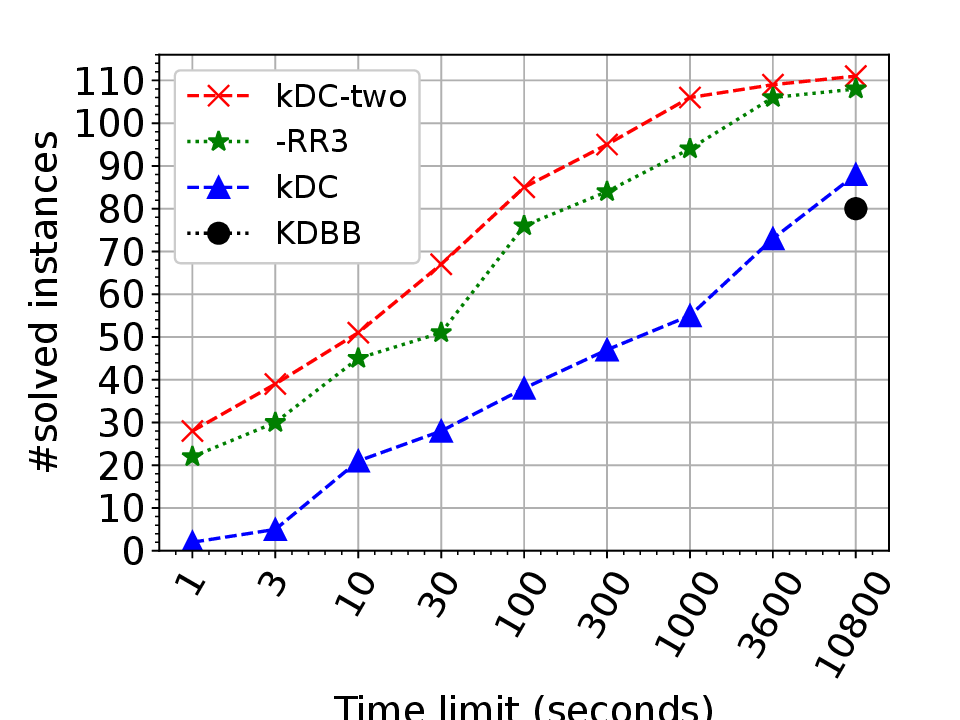}
}
\vspace{-2pt}
\caption{Number of solved instances by varying time limit for Facebook graphs (best viewed in color)}
\label{fig:facebook}
\end{minipage}
\end{figure*}

\section{Experiments}
\label{sec:experiment}

\newcommand{\wothree}{\kw{\text{-}RR3}}
\newcommand{\wofour}{\kw{\text{-}RR4}}
\newcommand{\woboth}{\kw{\text{-}RR3\&4}}

In this section, we evaluate the practical performance of \kdct, by
comparing it against the following two existing algorithms.
\begin{itemize}
	\item \kdc: the state-of-the-art algorithm proposed
		in~\cite{SIGMOD24:Chang}.
	\item \kdbb: the existing algorithm proposed in~\cite{AAAI22:Gao}.
\end{itemize}
We implemented our algorithm \kdct based on the code-base of \kdc that
is downloaded from \url{https://lijunchang.github.io/Maximum-kDC/};
thus, all optimizations/techniques that are implemented in \kdc, except
the second-order reduction rule, are used in \kdct.
In addition, we also implemented the following variant of \kdct
to evaluate the effectiveness of our new reduction rule {\bf RR3}.
\begin{itemize}
	\item \wothree: \kdct without the reduction rule {\bf RR3}.
\end{itemize}
All algorithms are implemented in C++ and compiled with the \mbox{-O3} flag.
All experiments are conducted in single-thread modes on a machine with an
Intel Core i7-8700 CPU and 64GB main memory.

We run the algorithms on the following three graph
collections, which are the same ones tested
in~\cite{AAAI22:Gao,SIGMOD24:Chang}.
\begin{itemize}
	\item The {\bf real-world graphs}
		collection~\footnote{\url{http://lcs.ios.ac.cn/~caisw/Resource/realworld\%20graphs.tar.gz}}
		contains $139$ real-world graphs from the Network Data Repository
		with up to $5.87\times 10^7$ vertices and $1.06\times 10^8$
		edges.
	\item The {\bf Facebook graphs} collection~\footnote{\url{https://networkrepository.com/socfb.php}}
		contains $114$ Facebook social networks from the Network Data
		Repository with up to $5.92\times 10^7$ vertices and $9.25\times
		10^7$ edges.
	\item The {\bf DIMACS10\&SNAP graphs} collection contains
			37 graphs with
			up to $1.04\times 10^6$ vertices and $6.89\times 10^6$
			edges. 
			$27$ of them are from 
			DIMACS10~\footnote{\url{https://www.cc.gatech.edu/dimacs10/downloads.shtml}}
			and $10$ graphs are from
			SNAP~\footnote{\url{http://snap.stanford.edu/data/}}.
\end{itemize}

Same as~\cite{AAAI22:Gao,SIGMOD24:Chang}, we choose $k$ from
$\{1,3,5,10,15,20\}$, and set a time limit of $3$ hours for each
testing (\ie, running a particular algorithm on a specific graph with a
chosen $k$ value).

\subsection{Against the Existing Algorithms}

In this subsection, we evaluate the efficiency of our algorithm \kdct
against the existing algorithms \kdc and \kdbb.
Note that, (1)~\kdc is the fastest existing algorithm, and (2)~as the
code of \kdbb is not available, the results of \kdbb reported in this
subsection are obtained from the original paper of \kdbb~\cite{AAAI22:Gao}.

\begin{table}[htb]
\footnotesize
\setlength{\tabcolsep}{2pt}
\centering
\caption{Number of instances solved by the algorithms \kdct, \kdc and \kdbb
	with a time limit of $3$ hours (best
performers are highlighted in bold)}
\label{table:against_existing}
\begin{tabular}{c|ccc|ccc|ccc}
\hline
& \multicolumn{3}{c|}{Real-world graphs} & \multicolumn{3}{c|}{Facebook
graphs} & \multicolumn{3}{c}{DIMACS10\&SNAP} \\ 
& \kdct & \kdc & \kdbb & \kdct & \kdc & \kdbb & \kdct & \kdc & \kdbb \\ \hline
$k = 1$ & {\bf 137} & 133 & 117 & {\bf 114} & {\bf 114} & 110 & {\bf 37} & {\bf 37} & 36 \\
$k=3$ & {\bf 136} & 130 & 107 & {\bf 114} & {\bf 114} & 110 & {\bf 37} & {\bf 37} & 35 \\
$k=5$ & {\bf 136} & 127 & 104 & {\bf 114} & {\bf 114} & 108 & {\bf 37} & {\bf 37} & 34 \\
$k=10$ & {\bf 128} & 119 & 85 & {\bf 112} & 111 & 109 & {\bf 36} & {\bf 36} & 30 \\
$k=15$ & {\bf 126} & 110 & 68 & {\bf 112} & 101 & 103 & {\bf 36} & 29 & 25 \\
$k=20$ & {\bf 113} & 104 & 56 & {\bf 111} & 88 & 80 & {\bf 34} & 27 & 22 \\ \hline
\end{tabular}
\end{table}

We first report in \ctab\ref{table:against_existing} the total number of
graph instances that are solved by each algorithm with a time limit of 3
hours. We can see that for all three algorithms, the number of solved
instances decreases when $k$ increases; this indicates that when $k$
increases, the problem becomes more difficult to solve.
Nevertheless, our algorithm \kdct consistently outperforms the
two existing algorithms by solving more instances within
the time limit. The improvement is more profound when $k$ becomes
large. For example, for $k=15$, \kdct solves $16$, $11$ and $7$ more
instances than the fastest existing algorithm \kdc on the three graph
collections, respectively; for $k=20$, the numbers are $9$, $23$,
and $7$.

Secondly, we compare the number of instances solved by the algorithms
when varying the time limit from $1$ second to $3$ hours. The results on
the real-world graphs and Facebook graphs for different $k$ values are
shown in Figures~\ref{fig:realworld} and \ref{fig:facebook},
respectively. We can see that our algorithm \kdct consistently
outperforms \kdc across all the time limits.
Also notice that, on the real-world graphs collection, our algorithm
\kdct with a time
limit of $1$ second solves even more instances than \kdbb with a time
limit of $3$ hours.
We also remark that our algorithm \kdct solves all $114$ Facebook graphs
with a time limit of $30$ seconds for $k=1$, $3$ and $5$, while the time
limits needed by \kdc are $125$, $393$ and $1353$ seconds, respectively;
on the other hand, \kdbb is not able to solve all instances with a time
limit of $3$ hours.

\begin{table}[ht]
\centering
\footnotesize
\addtolength{\tabcolsep}{-2.5pt}
\caption{Processing time (in seconds) of \wothree, \kdct, \kdc,
		and \kdbb on the $39$ Facebook
		graphs with more than $15,000$ vertices.
		Best performers are highlighted
		in bold: specifically, if the running time is slower than the
		fastest running time by only less than 10\%, it is 
	considered as the best.}
\label{table:facebook_time}
\begin{tabular}{crr|cccc|ccc}
\hline
 & & & \multicolumn{4}{c|}{$k=10$} & \multicolumn{3}{c}{$k=15$} \\ 
 & $n$ & $m$ & \wothree & \kdct & \kdc & \kdbb & \wothree & \kdct & \kdc \\ \hline
A-anon & 3M & 23M & 66 & {\bf 31} & 73 & - & 627 & {\bf 239} & 1062\\
Auburn71 & 18K & 973K & 4.7 & {\bf 2.7} & 956 & 1195 & 148 & {\bf 53} & -\\
B-anon & 2M & 20M & 79 & 51 & {\bf 44} & - & 2780 & {\bf 1447} & 2129\\
Berkeley13 & 22K & 852K & 0.24 & {\bf 0.22} & 0.34 & 630 & 0.77 & {\bf 0.54} & 70\\
BU10 & 19K & 637K & 0.66 & {\bf 0.34} & 4.0 & 370 & 3.6 & {\bf 1.8} & 156\\
Cornell5 & 18K & 790K & 0.98 & {\bf 0.76} & 17 & 2636 & 7.4 & {\bf 2.7} & 228\\
FSU53 & 27K & 1M & 1.4 & {\bf 1.0} & 248 & 1400 & 54 & {\bf 34} & -\\
Harvard1 & 15K & 824K & 2.3 & {\bf 1.6} & 11 & 1354 & 13 & {\bf 4.0} & 341\\
Indiana & 29K & 1M & 1.9 & {\bf 1.1} & 19 & 1421 & 16 & {\bf 6.5} & 861\\
Indiana69 & 29K & 1M & 1.9 & {\bf 1.1} & 19 & 1321 & 16 & {\bf 6.5} & 858\\
Maryland58 & 20K & 744K & 0.25 & {\bf 0.13} & 0.60 & 239 & 0.24 & {\bf 0.15} & 0.98\\
Michigan23 & 30K & 1M & 0.88 & {\bf 0.79} & 2.2 & 1384 & 5.9 & {\bf 2.2} & 211\\
MSU24 & 32K & 1M & {\bf 0.41} & {\bf 0.38} & 0.47 & 879 & 0.84 & {\bf 0.56} & 4.8\\
MU78 & 15K & 649K & 1.5 & {\bf 0.63} & 67 & 306 & 7.5 & {\bf 3.4} & 2437\\
NYU9 & 21K & 715K & {\bf 0.11} & {\bf 0.10} & 0.12 & 466 & {\bf 0.14} & {\bf 0.13} & 0.58\\
Oklahoma97 & 17K & 892K & 1.9 & {\bf 1.5} & 379 & 6926 & 76 & {\bf 50} & -\\
OR & 63K & 816K & 4.1 & {\bf 1.3} & 55 & 1486 & 31 & {\bf 10.0} & 400\\
Penn94 & 41K & 1M & {\bf 0.26} & {\bf 0.25} & 0.29 & 1972 & 0.39 & {\bf 0.35} & 0.61\\
Rutgers89 & 24K & 784K & {\bf 0.11} & {\bf 0.10} & 0.20 & 386 & 0.23 & {\bf 0.13} & 4.2\\
Tennessee95 & 16K & 770K & {\bf 0.99} & {\bf 0.94} & 1.8 & 554 & 1.5 & {\bf 1.1} & 20\\
Texas80 & 31K & 1M & 3.0 & {\bf 2.0} & 80 & 753 & 13 & {\bf 6.7} & 912\\
Texas84 & 36K & 1M & 12 & {\bf 3.6} & 1321 & 10253 & 364 & {\bf 85} & -\\
UC33 & 16K & 522K & {\bf 0.09} & {\bf 0.08} & 0.14 & 263 & 0.21 & {\bf 0.17} & 1.1\\
UCLA & 20K & 747K & {\bf 0.11} & {\bf 0.11} & 0.14 & 290 & 0.20 & {\bf 0.17} & 0.58\\
UCLA26 & 20K & 747K & {\bf 0.11} & {\bf 0.11} & 0.14 & 288 & 0.20 & {\bf 0.17} & 0.55\\
UConn & 17K & 604K & 0.09 & {\bf 0.08} & 0.13 & 194 & 0.34 & {\bf 0.20} & 2.1\\
UConn91 & 17K & 604K & 0.09 & {\bf 0.08} & 0.13 & 208 & 0.34 & {\bf 0.20} & 2.2\\
UF & 35K & 1M & 1.6 & {\bf 1.3} & 27 & 2579 & 34 & {\bf 16} & 10778\\
UF21 & 35K & 1M & 1.6 & {\bf 1.3} & 27 & 2571 & 34 & {\bf 16} & 10765\\
UGA50 & 24K & 1M & 49 & {\bf 26} & 3318 & 6794 & 1874 & {\bf 671} & -\\
UIllinois & 30K & 1M & 2.0 & {\bf 1.5} & 3.6 & 1245 & 6.5 & {\bf 3.5} & 85\\
UIllinois20 & 30K & 1M & 2.0 & {\bf 1.6} & 3.6 & 1217 & 6.5 & {\bf 3.5} & 86\\
UMass92 & 16K & 519K & 0.21 & {\bf 0.19} & 0.30 & 318 & 0.65 & {\bf 0.43} & 5.1\\
UNC28 & 18K & 766K & 1.1 & {\bf 0.67} & 2.1 & 380 & 4.8 & {\bf 2.1} & 12\\
USC35 & 17K & 801K & {\bf 0.41} & {\bf 0.38} & 0.52 & 409 & 1.3 & {\bf 0.89} & 10\\
UVA16 & 17K & 789K & 2.4 & {\bf 1.4} & 14 & 552 & 25 & {\bf 10} & 394\\
Virginia63 & 21K & 698K & 0.30 & {\bf 0.27} & 1.1 & 215 & 0.46 & {\bf 0.39} & 5.6\\
Wisconsin87 & 23K & 835K & 2.0 & {\bf 1.0} & 47 & 924 & 21 & {\bf 8.3} & 1604\\
wosn-friends & 63K & 817K & 4.1 & {\bf 1.3} & 54 & 1260 & 31 & {\bf 10.0} & 401\\
\hline
\end{tabular}

\end{table}

Thirdly, we report the actual processing time of \kdct, \kdc and \kdbb
on a subset of Facebook graphs that have more than $15,000$ vertices
for $k=10$ and $15$ in \ctab\ref{table:facebook_time}, where `$-$'
indicates that the processing time is longer than the $3$-hour limit.
There are totally $41$ such graphs. But none of the tested algorithms
can finish within the time limit of $3$ hours on graphs \texttt{konect}
and \texttt{uci-uni}; thus, these two graphs are omitted from
\ctab\ref{table:facebook_time}. Also, the results of \kdbb for
$k=15$ are omitted, as they are not available.
The number of vertices and edges for each graph are illustrated in the
second and third columns of \ctab\ref{table:facebook_time}, respectively.
From \ctab\ref{table:facebook_time}, we can observe that our algorithm
\kdct consistently and significantly outperforms \kdc, which in turn
runs significantly faster than \kdbb, across all these graphs.
In particular, \kdct is on average $45\times$ and $102\times$
faster than \kdc for $k=10$ and $15$, respectively.

In summary, our algorithm \kdct consistently solves more graph instances
than the fastest existing algorithm \kdc when varying the time limit
from $1$ second to $3$ hours, and also consistently runs faster than
\kdc across the different graphs with an average speed up up-to two
orders of magnitude.

\subsection{Ablation Studies}

In this subsection, we first evaluate the effectiveness of our new reduction
rule {\bf RR3} by comparing \kdct with \wothree which is the variant of
\kdct without {\bf RR3}. The results of \wothree are also shown in
Figures~\ref{fig:realworld} and \ref{fig:facebook} and
\ctab\ref{table:facebook_time}. We can observe from
Figures~\ref{fig:realworld} and \ref{fig:facebook} that the reduction
rule {\bf RR3} enables \kdct to solve more graph instances across the
different time limits. In particular, \kdct solves $3$, $5$ and $3$ more
instances than \wothree for $k=10$, $15$ and $20$, respectively, on the
real-world graphs collection with the time limit of $3$ hours.
From \ctab\ref{table:facebook_time}, we can see that \kdct
consistently runs faster than \wothree across the different Facebook
graphs with an average speed up of $2$ times for $k=15$. This
demonstrates the practical effectiveness of our new reduction rule {\bf
RR3}.

\begin{table}[t]
\small
\setlength{\tabcolsep}{2pt}
\centering
\caption{Number of graphs with small maximum $k$-defective clique (\ie,
	$\omega_k(G) \leq k+1$) and number of graphs with large maximum
$k$-defective clique (\ie, $\omega_k(G) \geq k+2$)}
\label{table:small_large}
\begin{tabular}{c|cc|cc|cc}
\hline
& \multicolumn{2}{c|}{Real-world graphs} & \multicolumn{2}{c|}{Facebook
graphs} & \multicolumn{2}{c}{DIMACS10\&SNAP} \\ 
& \#small & \#large & \#small & \#large & \#small & \#large \\ \hline
$k = 1$ & 2 & 137 & 0 & 114 & 0 & 37 \\
$k=3$ & 13 & 126 & 0 & 114 & 0 & 37 \\
$k=5$ & 22 & 117 & 1 & 113 & 1 & 36 \\
$k=10$ & 40 & 98 & 1 & 111 & 8 & 29 \\
$k=15$ & 47 & 91 & 1 & 111 & 12 & 25 \\
$k=20$ & 53 & 83 & 1 & 111 & 16 & 21 \\ \hline
\end{tabular}
\end{table}

Secondly, we compare \wothree with \kdc. Note that, \wothree uses the
same branching rule, reduction rules, and upper bounds as \kdc; the only
difference between them is that \wothree conducts the computation in two
stages and exploits the diameter-two property for pruning in Stage-I.
From Figures~\ref{fig:realworld} and \ref{fig:facebook} and
\ctab\ref{table:facebook_time}, we can see that \wothree consistently
outperforms \kdc.
This demonstrates the practical effectiveness of diameter-two-based
pruning.
To gain more insights, we report in \ctab\ref{table:small_large} the
number of graphs with small maximum $k$-defective clique (\ie,
$\omega_k(G) \leq k+1$) and the number of graphs with large maximum
$k$-defective clique (\ie, $\omega_k(G) \geq k+2$) for each graph
collection and $k$ value. We can see that when $k$ increases, the
proportion of graphs with $\omega_k(G) \geq k+2$ decreases.
Nevertheless, even for $k=20$, there are still a lot of graphs with
$\omega_k(G) \geq k+2$ such that \kdct runs in $\bigo^*(
(\alpha\Delta)^{k+2} \gamma_{k-1}^\alpha)$ time; this is especially true
for Facebook graphs.

\section{Related Work}
\label{sec:related_work}

The concept of defective clique was firstly studied in~\cite{Bio06:Yu}
for predicting missing interactions between proteins in biological
networks.
Since then, designing exact algorithms for efficiently finding the
maximum defective clique has been investigated due to its importance,
despite being an NP-hard problem.
Early algorithms, such as those proposed
in~\cite{COA13:Trukhanov,DAM18:Gschwind,IJC21:Gschwind}, are inefficient
and can only deal with small graphs.
The \madec algorithm proposed by Chen et al.~\cite{COR21:Chen} is the
first algorithm that can handle large graphs, due to the incorporating
of a graph coloring-based upper bound and other pruning techniques.
The \kdbb algorithm proposed by Gao et al.~\cite{AAAI22:Gao} improves
the practical performance by proposing preprocessing as well as multiple
pruning techniques.
The \kdc algorithm proposed by Chang~\cite{SIGMOD24:Chang} is
the state-of-the-art algorithm which incorporates an improved
graph-coloring-based upper bound, a better initialization method, as
well as multiple new reduction rules.
From the theoretical perspective, among the existing algorithms, only
\madec~\cite{COR21:Chen} and \kdc~\cite{SIGMOD24:Chang} beats the
trivial time complexity of $\bigo^*(2^n)$.
Specifically, \madec runs in $\bigo^*(\gamma_{2k}^n)$ time while \kdc
improves the time complexity to $\bigo^*(\gamma_k^n)$, where $\gamma_k <
2$ is the largest real root of the equation $x^{k+3}-2x^{k+2}+1=0$ and
$k$ is the number of allowed missing edges.
In this paper, we proposed the \kdct algorithm to further improve both the
time complexity and practical performance for maximum defective clique
computation.

The problem of enumerating all maximal $k$-defective cliques was also
studied recently where the \pivotplus algorithm proposed
in~\cite{SIGMOD23:Dai} runs in
$\bigo^*(\gamma_k^n)$ time, the same as the time complexity of \kdc.
However, we remark that (1)~\pivotplus is
inefficient for finding the maximum $k$-defective clique in practice due to lack of
pruning techniques; (2)~\pivotplus achieves the time complexity 
via using a different branching technique from
\kdc and it is unclear how to improve the base of its time complexity
without changing the branching rule.
The problem of approximately counting $k$-defective cliques of a
particular size, for the special case of $k=1$ and $2$, was studied
in~\cite{WWW20:Jain}; however, the techniques of~\cite{WWW20:Jain}
cannot be used for finding the
maximum $k$-defective clique and for a general $k$.

Another related problem is maximum clique computation, as clique is a
special case of $k$-defective clique for $k=0$.
The maximum clique is not only NP-hard to compute
exactly~\cite{CCC72:Karp}, but also NP-hard to approximate within a
factor of $n^{1-\epsilon}$ for any constant $0 < \epsilon <
1$~\cite{FOCS96:Hastad}.
Nevertheless, extensive efforts have been spent on reducing the time
complexity from the trivial $\bigo^*(2^n)$ to $\bigo^*(1.4422^n)$,
$\bigo^*(1.2599^n)$~\cite{SJC77:Tarjan},
$\bigo^*(1.2346^n)$~\cite{TC86:Jian}, and
$\bigo^*(1.2108^n)$~\cite{JA86:Robson}, with the state of the art being
$\bigo^*(1.1888^n)$~\cite{clique}.
On the other hand, extensive efforts have also been spent on improving
the practical efficiency, while ignoring theoretical time complexity, 
\eg,~\cite{ORL90:Carraghan,COR17:Li,ICTAI13:Li,JGO94:Pardalos,IM15:Pattabiraman,JSC15:Rossi,COR16:Segundo,WALCOM17:Tomita,WALCOM10:Tomita,ICDE13:Xiang,KDD19:Chang}.
However, it is non-trivial to extend these techniques to $k$-defective
computation for a general $k$.

\section{Conclusion}
\label{sec:conclusion}

In this paper, we proposed the \kdct algorithm for exact maximum
$k$-defective clique computation that runs in $\bigo^*(
(\alpha\Delta)^{k+2} \times \gamma_{k-1}^\alpha)$ time for graphs with
$\omega_k(G)\geq k+2$, and in $\bigo^*(\gamma_{k-1}^n)$ time
otherwise. This improves the state-of-the-art time complexity
$\bigo^*(\gamma_k^n)$.
We also proved that \kdct, with slight modification, runs in $\bigo^*(
(\alpha\Delta)^{k+2} \times (k+1)^{\alpha+k+1-\omega_k(G)})$ time when
using the degeneracy gap $\alpha+k+1-\omega_k(G)$ parameterization. In
addition, from the practical side, we designed a new
degree-sequence-based reduction rule that can be conducted in linear
time, and theoretically demonstrated its effectiveness compared with
other reduction rules.
Extensive empirical studies on three benchmark collections with $290$
graphs in total showed that \kdct outperforms the existing fastest
algorithm by several orders of magnitude in practice.

\balance
\bibliographystyle{ACM-Reference-Format}
\bibliography{sigproc}


\begin{thebibliography}{47}


\ifx \showCODEN    \undefined \def \showCODEN     #1{\unskip}     \fi
\ifx \showDOI      \undefined \def \showDOI       #1{#1}\fi
\ifx \showISBNx    \undefined \def \showISBNx     #1{\unskip}     \fi
\ifx \showISBNxiii \undefined \def \showISBNxiii  #1{\unskip}     \fi
\ifx \showISSN     \undefined \def \showISSN      #1{\unskip}     \fi
\ifx \showLCCN     \undefined \def \showLCCN      #1{\unskip}     \fi
\ifx \shownote     \undefined \def \shownote      #1{#1}          \fi
\ifx \showarticletitle \undefined \def \showarticletitle #1{#1}   \fi
\ifx \showURL      \undefined \def \showURL       {\relax}        \fi
\providecommand\bibfield[2]{#2}
\providecommand\bibinfo[2]{#2}
\providecommand\natexlab[1]{#1}
\providecommand\showeprint[2][]{arXiv:#2}

\bibitem[Abello et~al\mbox{.}(2002)]%
        {LATIN02:Abello}
\bibfield{author}{\bibinfo{person}{James Abello}, \bibinfo{person}{Mauricio
  G.~C. Resende}, {and} \bibinfo{person}{Sandra Sudarsky}.}
  \bibinfo{year}{2002}\natexlab{}.
\newblock \showarticletitle{Massive Quasi-Clique Detection}. In
  \bibinfo{booktitle}{\emph{Proc. of {LATIN}'02}}
  \emph{(\bibinfo{series}{Lecture Notes in Computer Science},
  Vol.~\bibinfo{volume}{2286})}. \bibinfo{publisher}{Springer},
  \bibinfo{pages}{598--612}.
\newblock


\bibitem[Ahmed et~al\mbox{.}(2016)]%
        {ahmed2016survey}
\bibfield{author}{\bibinfo{person}{Mohiuddin Ahmed},
  \bibinfo{person}{Abdun~Naser Mahmood}, {and} \bibinfo{person}{Md~Rafiqul
  Islam}.} \bibinfo{year}{2016}\natexlab{}.
\newblock \showarticletitle{A survey of anomaly detection techniques in
  financial domain}.
\newblock \bibinfo{journal}{\emph{Future Generation Computer Systems}}
  \bibinfo{volume}{55} (\bibinfo{year}{2016}), \bibinfo{pages}{278--288}.
\newblock


\bibitem[Balasundaram et~al\mbox{.}(2011)]%
        {OR11:Balasundaram}
\bibfield{author}{\bibinfo{person}{Balabhaskar Balasundaram},
  \bibinfo{person}{Sergiy Butenko}, {and} \bibinfo{person}{Illya~V. Hicks}.}
  \bibinfo{year}{2011}\natexlab{}.
\newblock \showarticletitle{Clique Relaxations in Social Network Analysis: The
  Maximum \emph{k}-Plex Problem}.
\newblock \bibinfo{journal}{\emph{Operations Research}} \bibinfo{volume}{59},
  \bibinfo{number}{1} (\bibinfo{year}{2011}), \bibinfo{pages}{133--142}.
\newblock


\bibitem[Bedi and Sharma(2016)]%
        {bedi2016community}
\bibfield{author}{\bibinfo{person}{Punam Bedi} {and} \bibinfo{person}{Chhavi
  Sharma}.} \bibinfo{year}{2016}\natexlab{}.
\newblock \showarticletitle{Community detection in social networks}.
\newblock \bibinfo{journal}{\emph{Wiley Interdisciplinary Reviews: Data Mining
  and Knowledge Discovery}} \bibinfo{volume}{6}, \bibinfo{number}{3}
  (\bibinfo{year}{2016}), \bibinfo{pages}{115--135}.
\newblock


\bibitem[Bourjolly et~al\mbox{.}(2002)]%
        {EJOR02:Bourjolly}
\bibfield{author}{\bibinfo{person}{Jean{-}Marie Bourjolly},
  \bibinfo{person}{Gilbert Laporte}, {and} \bibinfo{person}{Gilles Pesant}.}
  \bibinfo{year}{2002}\natexlab{}.
\newblock \showarticletitle{An exact algorithm for the maximum k-club problem
  in an undirected graph}.
\newblock \bibinfo{journal}{\emph{Eur. J. Oper. Res.}} \bibinfo{volume}{138},
  \bibinfo{number}{1} (\bibinfo{year}{2002}), \bibinfo{pages}{21--28}.
\newblock


\bibitem[Carraghan and Pardalos(1990)]%
        {ORL90:Carraghan}
\bibfield{author}{\bibinfo{person}{Randy Carraghan} {and}
  \bibinfo{person}{Panos~M. Pardalos}.} \bibinfo{year}{1990}\natexlab{}.
\newblock \showarticletitle{An Exact Algorithm for the Maximum Clique Problem}.
\newblock \bibinfo{journal}{\emph{Oper. Res. Lett.}} \bibinfo{volume}{9},
  \bibinfo{number}{6} (\bibinfo{date}{Nov.} \bibinfo{year}{1990}),
  \bibinfo{pages}{375--382}.
\newblock
\showISSN{0167-6377}


\bibitem[Chang(2019)]%
        {KDD19:Chang}
\bibfield{author}{\bibinfo{person}{Lijun Chang}.}
  \bibinfo{year}{2019}\natexlab{}.
\newblock \showarticletitle{Efficient Maximum Clique Computation over Large
  Sparse Graphs}. In \bibinfo{booktitle}{\emph{Proc. of KDD'19}}.
  \bibinfo{pages}{529--538}.
\newblock


\bibitem[Chang(2023)]%
        {SIGMOD24:Chang}
\bibfield{author}{\bibinfo{person}{Lijun Chang}.}
  \bibinfo{year}{2023}\natexlab{}.
\newblock \showarticletitle{Efficient Maximum K-Defective Clique Computation
  with Improved Time Complexity}.
\newblock \bibinfo{journal}{\emph{Proc. ACM Manag. Data}} \bibinfo{volume}{1},
  \bibinfo{number}{3}, Article \bibinfo{articleno}{209} (\bibinfo{date}{nov}
  \bibinfo{year}{2023}), \bibinfo{numpages}{26}~pages.
\newblock
\urldef\tempurl%
\url{https://doi.org/10.1145/3617313}
\showDOI{\tempurl}


\bibitem[Chang and Qin(2018)]%
        {Book18:Chang}
\bibfield{author}{\bibinfo{person}{Lijun Chang} {and} \bibinfo{person}{Lu
  Qin}.} \bibinfo{year}{2018}\natexlab{}.
\newblock \bibinfo{booktitle}{\emph{Cohesive Subgraph Computation over Large
  Sparse Graphs}}.
\newblock \bibinfo{publisher}{Springer Series in the Data Sciences}.
\newblock


\bibitem[Chen et~al\mbox{.}(2021)]%
        {COR21:Chen}
\bibfield{author}{\bibinfo{person}{Xiaoyu Chen}, \bibinfo{person}{Yi Zhou},
  \bibinfo{person}{Jin{-}Kao Hao}, {and} \bibinfo{person}{Mingyu Xiao}.}
  \bibinfo{year}{2021}\natexlab{}.
\newblock \showarticletitle{Computing maximum k-defective cliques in massive
  graphs}.
\newblock \bibinfo{journal}{\emph{Comput. Oper. Res.}}  \bibinfo{volume}{127}
  (\bibinfo{year}{2021}), \bibinfo{pages}{105131}.
\newblock


\bibitem[Cheng et~al\mbox{.}(2011)]%
        {TODS11:Cheng}
\bibfield{author}{\bibinfo{person}{James Cheng}, \bibinfo{person}{Yiping Ke},
  \bibinfo{person}{Ada~Wai{-}Chee Fu}, \bibinfo{person}{Jeffrey~Xu Yu}, {and}
  \bibinfo{person}{Linhong Zhu}.} \bibinfo{year}{2011}\natexlab{}.
\newblock \showarticletitle{Finding maximal cliques in massive networks}.
\newblock \bibinfo{journal}{\emph{{ACM} Trans. Database Syst.}}
  \bibinfo{volume}{36}, \bibinfo{number}{4} (\bibinfo{year}{2011}),
  \bibinfo{pages}{21:1--21:34}.
\newblock


\bibitem[Dai et~al\mbox{.}(2023)]%
        {SIGMOD23:Dai}
\bibfield{author}{\bibinfo{person}{Qiangqiang Dai}, \bibinfo{person}{Rong{-}Hua
  Li}, \bibinfo{person}{Meihao Liao}, {and} \bibinfo{person}{Guoren Wang}.}
  \bibinfo{year}{2023}\natexlab{}.
\newblock \showarticletitle{Maximal Defective Clique Enumeration}.
\newblock \bibinfo{journal}{\emph{Proc. {ACM} Manag. Data}}
  \bibinfo{volume}{1}, \bibinfo{number}{1} (\bibinfo{year}{2023}),
  \bibinfo{pages}{77:1--77:26}.
\newblock
\urldef\tempurl%
\url{https://doi.org/10.1145/3588931}
\showDOI{\tempurl}


\bibitem[Eppstein et~al\mbox{.}(2013)]%
        {JEA13:Eppstein}
\bibfield{author}{\bibinfo{person}{David Eppstein}, \bibinfo{person}{Maarten
  L{\"{o}}ffler}, {and} \bibinfo{person}{Darren Strash}.}
  \bibinfo{year}{2013}\natexlab{}.
\newblock \showarticletitle{Listing All Maximal Cliques in Large Sparse
  Real-World Graphs}.
\newblock \bibinfo{journal}{\emph{{ACM} Journal of Experimental Algorithmics}}
  \bibinfo{volume}{18} (\bibinfo{year}{2013}).
\newblock


\bibitem[Feige et~al\mbox{.}(2001)]%
        {Algorithmica01:Feige}
\bibfield{author}{\bibinfo{person}{Uriel Feige}, \bibinfo{person}{Guy
  Kortsarz}, {and} \bibinfo{person}{David Peleg}.}
  \bibinfo{year}{2001}\natexlab{}.
\newblock \showarticletitle{The Dense \emph{k}-Subgraph Problem}.
\newblock \bibinfo{journal}{\emph{Algorithmica}} \bibinfo{volume}{29},
  \bibinfo{number}{3} (\bibinfo{year}{2001}), \bibinfo{pages}{410--421}.
\newblock
\urldef\tempurl%
\url{https://doi.org/10.1007/S004530010050}
\showDOI{\tempurl}


\bibitem[Fomin and Kratsch(2010)]%
        {Book2010:Fomin}
\bibfield{author}{\bibinfo{person}{Fedor~V. Fomin} {and}
  \bibinfo{person}{Dieter Kratsch}.} \bibinfo{year}{2010}\natexlab{}.
\newblock \bibinfo{booktitle}{\emph{Exact Exponential Algorithms}}.
\newblock \bibinfo{publisher}{Springer}.
\newblock


\bibitem[Gao et~al\mbox{.}(2022)]%
        {AAAI22:Gao}
\bibfield{author}{\bibinfo{person}{Jian Gao}, \bibinfo{person}{Zhenghang Xu},
  \bibinfo{person}{Ruizhi Li}, {and} \bibinfo{person}{Minghao Yin}.}
  \bibinfo{year}{2022}\natexlab{}.
\newblock \showarticletitle{An Exact Algorithm with New Upper Bounds for the
  Maximum k-Defective Clique Problem in Massive Sparse Graphs}. In
  \bibinfo{booktitle}{\emph{Proc. of AAAI'22}}. \bibinfo{pages}{10174--10183}.
\newblock


\bibitem[Gschwind et~al\mbox{.}(2021)]%
        {IJC21:Gschwind}
\bibfield{author}{\bibinfo{person}{Timo Gschwind}, \bibinfo{person}{Stefan
  Irnich}, \bibinfo{person}{Fabio Furini}, {and}
  \bibinfo{person}{Roberto~Wolfler Calvo}.} \bibinfo{year}{2021}\natexlab{}.
\newblock \showarticletitle{A Branch-and-Price Framework for Decomposing Graphs
  into Relaxed Cliques}.
\newblock \bibinfo{journal}{\emph{{INFORMS} J. Comput.}} \bibinfo{volume}{33},
  \bibinfo{number}{3} (\bibinfo{year}{2021}), \bibinfo{pages}{1070--1090}.
\newblock


\bibitem[Gschwind et~al\mbox{.}(2018)]%
        {DAM18:Gschwind}
\bibfield{author}{\bibinfo{person}{Timo Gschwind}, \bibinfo{person}{Stefan
  Irnich}, {and} \bibinfo{person}{Isabel Podlinski}.}
  \bibinfo{year}{2018}\natexlab{}.
\newblock \showarticletitle{Maximum weight relaxed cliques and Russian Doll
  Search revisited}.
\newblock \bibinfo{journal}{\emph{Discret. Appl. Math.}}  \bibinfo{volume}{234}
  (\bibinfo{year}{2018}), \bibinfo{pages}{131--138}.
\newblock


\bibitem[H{\aa}stad(1996)]%
        {FOCS96:Hastad}
\bibfield{author}{\bibinfo{person}{Johan H{\aa}stad}.}
  \bibinfo{year}{1996}\natexlab{}.
\newblock \showarticletitle{Clique is Hard to Approximate Within
  n\({}^{\mbox{1-epsilon}}\)}. In \bibinfo{booktitle}{\emph{Proc. of FOCS'96}}.
  \bibinfo{pages}{627--636}.
\newblock


\bibitem[Jain and Seshadhri(2020a)]%
        {WSDM20:Jain}
\bibfield{author}{\bibinfo{person}{Shweta Jain} {and} \bibinfo{person}{C.
  Seshadhri}.} \bibinfo{year}{2020}\natexlab{a}.
\newblock \showarticletitle{The Power of Pivoting for Exact Clique Counting}.
  In \bibinfo{booktitle}{\emph{Proc. {WSDM}'20}}. \bibinfo{publisher}{{ACM}},
  \bibinfo{pages}{268--276}.
\newblock


\bibitem[Jain and Seshadhri(2020b)]%
        {WWW20:Jain}
\bibfield{author}{\bibinfo{person}{Shweta Jain} {and} \bibinfo{person}{C.
  Seshadhri}.} \bibinfo{year}{2020}\natexlab{b}.
\newblock \showarticletitle{Provably and Efficiently Approximating Near-cliques
  using the Tur{\'{a}}n Shadow: {PEANUTS}}. In \bibinfo{booktitle}{\emph{Proc.
  of {WWW}'20}}. \bibinfo{publisher}{{ACM} / {IW3C2}},
  \bibinfo{pages}{1966--1976}.
\newblock


\bibitem[Jian(1986)]%
        {TC86:Jian}
\bibfield{author}{\bibinfo{person}{Tang Jian}.}
  \bibinfo{year}{1986}\natexlab{}.
\newblock \showarticletitle{An \emph{O}(2\({}^{\mbox{0.304\emph{n}}}\))
  Algorithm for Solving Maximum Independent Set Problem}.
\newblock \bibinfo{journal}{\emph{{IEEE} Trans. Computers}}
  \bibinfo{volume}{35}, \bibinfo{number}{9} (\bibinfo{year}{1986}),
  \bibinfo{pages}{847--851}.
\newblock


\bibitem[Karp(1972)]%
        {CCC72:Karp}
\bibfield{author}{\bibinfo{person}{Richard~M. Karp}.}
  \bibinfo{year}{1972}\natexlab{}.
\newblock \showarticletitle{Reducibility Among Combinatorial Problems}. In
  \bibinfo{booktitle}{\emph{Proc. of CCC'72}}. \bibinfo{pages}{85--103}.
\newblock


\bibitem[Lee et~al\mbox{.}(2010)]%
        {MMGD10:Lee}
\bibfield{author}{\bibinfo{person}{Victor~E. Lee}, \bibinfo{person}{Ning Ruan},
  \bibinfo{person}{Ruoming Jin}, {and} \bibinfo{person}{Charu~C. Aggarwal}.}
  \bibinfo{year}{2010}\natexlab{}.
\newblock \showarticletitle{A Survey of Algorithms for Dense Subgraph
  Discovery}.
\newblock In \bibinfo{booktitle}{\emph{Managing and Mining Graph Data}}.
  \bibinfo{series}{Advances in Database Systems}, Vol.~\bibinfo{volume}{40}.
  \bibinfo{publisher}{Springer}, \bibinfo{pages}{303--336}.
\newblock


\bibitem[Li et~al\mbox{.}(2013)]%
        {ICTAI13:Li}
\bibfield{author}{\bibinfo{person}{Chu{-}Min Li}, \bibinfo{person}{Zhiwen
  Fang}, {and} \bibinfo{person}{Ke Xu}.} \bibinfo{year}{2013}\natexlab{}.
\newblock \showarticletitle{Combining MaxSAT Reasoning and Incremental Upper
  Bound for the Maximum Clique Problem}. In \bibinfo{booktitle}{\emph{Proc. of
  ICTAI'13}}.
\newblock


\bibitem[Li et~al\mbox{.}(2017)]%
        {COR17:Li}
\bibfield{author}{\bibinfo{person}{Chu{-}Min Li}, \bibinfo{person}{Hua Jiang},
  {and} \bibinfo{person}{Felip Many{\`{a}}}.} \bibinfo{year}{2017}\natexlab{}.
\newblock \showarticletitle{On minimization of the number of branches in
  branch-and-bound algorithms for the maximum clique problem}.
\newblock \bibinfo{journal}{\emph{Computers {\&} {OR}}}  \bibinfo{volume}{84}
  (\bibinfo{year}{2017}), \bibinfo{pages}{1--15}.
\newblock


\bibitem[Li et~al\mbox{.}(2020)]%
        {PVLDB20:Li}
\bibfield{author}{\bibinfo{person}{Ronghua Li}, \bibinfo{person}{Sen Gao},
  \bibinfo{person}{Lu Qin}, \bibinfo{person}{Guoren Wang},
  \bibinfo{person}{Weihua Yang}, {and} \bibinfo{person}{Jeffrey~Xu Yu}.}
  \bibinfo{year}{2020}\natexlab{}.
\newblock \showarticletitle{Ordering Heuristics for k-clique Listing}.
\newblock \bibinfo{journal}{\emph{Proc. {VLDB} Endow.}} \bibinfo{volume}{13},
  \bibinfo{number}{11} (\bibinfo{year}{2020}), \bibinfo{pages}{2536--2548}.
\newblock


\bibitem[Matula and Beck(1983)]%
        {JACM83:Matula}
\bibfield{author}{\bibinfo{person}{David~W. Matula} {and}
  \bibinfo{person}{Leland~L. Beck}.} \bibinfo{year}{1983}\natexlab{}.
\newblock \showarticletitle{Smallest-Last Ordering and clustering and Graph
  Coloring Algorithms}.
\newblock \bibinfo{journal}{\emph{J. {ACM}}} \bibinfo{volume}{30},
  \bibinfo{number}{3} (\bibinfo{year}{1983}), \bibinfo{pages}{417--427}.
\newblock


\bibitem[Mohri et~al\mbox{.}(2012)]%
        {Book12:Mohri}
\bibfield{author}{\bibinfo{person}{Mehryar Mohri}, \bibinfo{person}{Afshin
  Rostamizadeh}, {and} \bibinfo{person}{Ameet Talwalkar}.}
  \bibinfo{year}{2012}\natexlab{}.
\newblock \bibinfo{booktitle}{\emph{Foundations of Machine Learning}}.
\newblock \bibinfo{publisher}{{MIT} Press}.
\newblock


\bibitem[Pardalos and Xue(1994)]%
        {JGO94:Pardalos}
\bibfield{author}{\bibinfo{person}{Panos~M. Pardalos} {and}
  \bibinfo{person}{Jue Xue}.} \bibinfo{year}{1994}\natexlab{}.
\newblock \showarticletitle{The maximum clique problem}.
\newblock \bibinfo{journal}{\emph{J. global Optimization}} \bibinfo{volume}{4},
  \bibinfo{number}{3} (\bibinfo{year}{1994}), \bibinfo{pages}{301--328}.
\newblock


\bibitem[Pattabiraman et~al\mbox{.}(2015)]%
        {IM15:Pattabiraman}
\bibfield{author}{\bibinfo{person}{Bharath Pattabiraman}, \bibinfo{person}{Md.
  Mostofa~Ali Patwary}, \bibinfo{person}{Assefaw~Hadish Gebremedhin},
  \bibinfo{person}{Wei{-}keng Liao}, {and} \bibinfo{person}{Alok~N.
  Choudhary}.} \bibinfo{year}{2015}\natexlab{}.
\newblock \showarticletitle{Fast Algorithms for the Maximum Clique Problem on
  Massive Graphs with Applications to Overlapping Community Detection}.
\newblock \bibinfo{journal}{\emph{Internet Mathematics}} \bibinfo{volume}{11},
  \bibinfo{number}{4-5} (\bibinfo{year}{2015}), \bibinfo{pages}{421--448}.
\newblock


\bibitem[Pattillo et~al\mbox{.}(2013)]%
        {EOR13:Pattillo}
\bibfield{author}{\bibinfo{person}{Jeffrey Pattillo}, \bibinfo{person}{Nataly
  Youssef}, {and} \bibinfo{person}{Sergiy Butenko}.}
  \bibinfo{year}{2013}\natexlab{}.
\newblock \showarticletitle{On clique relaxation models in network analysis}.
\newblock \bibinfo{journal}{\emph{Eur. J. Oper. Res.}} \bibinfo{volume}{226},
  \bibinfo{number}{1} (\bibinfo{year}{2013}), \bibinfo{pages}{9--18}.
\newblock


\bibitem[Robson(1986)]%
        {JA86:Robson}
\bibfield{author}{\bibinfo{person}{J.~M. Robson}.}
  \bibinfo{year}{1986}\natexlab{}.
\newblock \showarticletitle{Algorithms for Maximum Independent Sets}.
\newblock \bibinfo{journal}{\emph{J. Algorithms}} \bibinfo{volume}{7},
  \bibinfo{number}{3} (\bibinfo{year}{1986}), \bibinfo{pages}{425--440}.
\newblock


\bibitem[Robson(2001)]%
        {clique}
\bibfield{author}{\bibinfo{person}{J.~M. Robson}.}
  \bibinfo{year}{2001}\natexlab{}.
\newblock \bibinfo{title}{Finding a maximum independent set in time
  $O(2^{n/4})$}.
\newblock
  \bibinfo{howpublished}{\small\url{https://www.labri.fr/perso/robson/mis/techrep.html}}.
\newblock


\bibitem[Rossi et~al\mbox{.}(2015)]%
        {JSC15:Rossi}
\bibfield{author}{\bibinfo{person}{Ryan~A. Rossi}, \bibinfo{person}{David~F.
  Gleich}, {and} \bibinfo{person}{Assefaw~Hadish Gebremedhin}.}
  \bibinfo{year}{2015}\natexlab{}.
\newblock \showarticletitle{Parallel Maximum Clique Algorithms with
  Applications to Network Analysis}.
\newblock \bibinfo{journal}{\emph{{SIAM} J. Scientific Computing}}
  \bibinfo{volume}{37}, \bibinfo{number}{5} (\bibinfo{year}{2015}).
\newblock


\bibitem[Sachs(1963)]%
        {JLMS63:Sachs}
\bibfield{author}{\bibinfo{person}{H. Sachs}.} \bibinfo{year}{1963}\natexlab{}.
\newblock \showarticletitle{{Regular Graphs with Given Girth and Restricted
  Circuits}}.
\newblock \bibinfo{journal}{\emph{Journal of the London Mathematical Society}}
  \bibinfo{volume}{s1-38}, \bibinfo{number}{1} (\bibinfo{year}{1963}),
  \bibinfo{pages}{423--429}.
\newblock


\bibitem[Segundo et~al\mbox{.}(2016)]%
        {COR16:Segundo}
\bibfield{author}{\bibinfo{person}{Pablo~San Segundo}, \bibinfo{person}{Alvaro
  Lopez}, {and} \bibinfo{person}{Panos~M. Pardalos}.}
  \bibinfo{year}{2016}\natexlab{}.
\newblock \showarticletitle{A new exact maximum clique algorithm for large and
  massive sparse graphs}.
\newblock \bibinfo{journal}{\emph{Computers \& Operations Research}}
  \bibinfo{volume}{66} (\bibinfo{year}{2016}), \bibinfo{pages}{81--94}.
\newblock


\bibitem[Sherali et~al\mbox{.}(2002)]%
        {TS02:Sherali}
\bibfield{author}{\bibinfo{person}{Hanif~D. Sherali}, \bibinfo{person}{J.~Cole
  Smith}, {and} \bibinfo{person}{Antonio~A. Trani}.}
  \bibinfo{year}{2002}\natexlab{}.
\newblock \showarticletitle{An Airspace Planning Model for Selecting
  Flight-plans Under Workload, Safety, and Equity Considerations}.
\newblock \bibinfo{journal}{\emph{Transp. Sci.}} \bibinfo{volume}{36},
  \bibinfo{number}{4} (\bibinfo{year}{2002}), \bibinfo{pages}{378--397}.
\newblock


\bibitem[Stozhkov et~al\mbox{.}(2022)]%
        {MP22:Stozhkov}
\bibfield{author}{\bibinfo{person}{Vladimir Stozhkov}, \bibinfo{person}{Austin
  Buchanan}, \bibinfo{person}{Sergiy Butenko}, {and} \bibinfo{person}{Vladimir
  Boginski}.} \bibinfo{year}{2022}\natexlab{}.
\newblock \showarticletitle{Continuous cubic formulations for cluster detection
  problems in networks}.
\newblock \bibinfo{journal}{\emph{Math. Program.}} \bibinfo{volume}{196},
  \bibinfo{number}{1} (\bibinfo{year}{2022}), \bibinfo{pages}{279--307}.
\newblock


\bibitem[Suratanee et~al\mbox{.}(2014)]%
        {suratanee2014characterizing}
\bibfield{author}{\bibinfo{person}{Apichat Suratanee},
  \bibinfo{person}{Martin~H Schaefer}, \bibinfo{person}{Matthew~J Betts},
  \bibinfo{person}{Zita Soons}, \bibinfo{person}{Heiko Mannsperger},
  \bibinfo{person}{Nathalie Harder}, \bibinfo{person}{Marcus Oswald},
  \bibinfo{person}{Markus Gipp}, \bibinfo{person}{Ellen Ramminger},
  \bibinfo{person}{Guillermo Marcus}, {et~al\mbox{.}}}
  \bibinfo{year}{2014}\natexlab{}.
\newblock \showarticletitle{Characterizing protein interactions employing a
  genome-wide siRNA cellular phenotyping screen}.
\newblock \bibinfo{journal}{\emph{PLoS computational biology}}
  \bibinfo{volume}{10}, \bibinfo{number}{9} (\bibinfo{year}{2014}),
  \bibinfo{pages}{e1003814}.
\newblock


\bibitem[Tarjan and Trojanowski(1977)]%
        {SJC77:Tarjan}
\bibfield{author}{\bibinfo{person}{Robert~Endre Tarjan} {and}
  \bibinfo{person}{Anthony~E. Trojanowski}.} \bibinfo{year}{1977}\natexlab{}.
\newblock \showarticletitle{Finding a Maximum Independent Set}.
\newblock \bibinfo{journal}{\emph{{SIAM} J. Comput.}} \bibinfo{volume}{6},
  \bibinfo{number}{3} (\bibinfo{year}{1977}), \bibinfo{pages}{537--546}.
\newblock


\bibitem[Tomita(2017)]%
        {WALCOM17:Tomita}
\bibfield{author}{\bibinfo{person}{Etsuji Tomita}.}
  \bibinfo{year}{2017}\natexlab{}.
\newblock \showarticletitle{Efficient Algorithms for Finding Maximum and
  Maximal Cliques and Their Applications}. In \bibinfo{booktitle}{\emph{Proc.
  of WALCOM'17}}. \bibinfo{pages}{3--15}.
\newblock


\bibitem[Tomita et~al\mbox{.}(2010)]%
        {WALCOM10:Tomita}
\bibfield{author}{\bibinfo{person}{Etsuji Tomita}, \bibinfo{person}{Yoichi
  Sutani}, \bibinfo{person}{Takanori Higashi}, \bibinfo{person}{Shinya
  Takahashi}, {and} \bibinfo{person}{Mitsuo Wakatsuki}.}
  \bibinfo{year}{2010}\natexlab{}.
\newblock \showarticletitle{A simple and faster branch-and-bound algorithm for
  finding a maximum clique}. In \bibinfo{booktitle}{\emph{Proc. of WALCOM'10}}.
  \bibinfo{pages}{191--203}.
\newblock


\bibitem[Trukhanov et~al\mbox{.}(2013)]%
        {COA13:Trukhanov}
\bibfield{author}{\bibinfo{person}{Svyatoslav Trukhanov},
  \bibinfo{person}{Chitra Balasubramaniam}, \bibinfo{person}{Balabhaskar
  Balasundaram}, {and} \bibinfo{person}{Sergiy Butenko}.}
  \bibinfo{year}{2013}\natexlab{}.
\newblock \showarticletitle{Algorithms for detecting optimal hereditary
  structures in graphs, with application to clique relaxations}.
\newblock \bibinfo{journal}{\emph{Comput. Optim. Appl.}} \bibinfo{volume}{56},
  \bibinfo{number}{1} (\bibinfo{year}{2013}), \bibinfo{pages}{113--130}.
\newblock


\bibitem[Xiang et~al\mbox{.}(2013)]%
        {ICDE13:Xiang}
\bibfield{author}{\bibinfo{person}{Jingen Xiang}, \bibinfo{person}{Cong Guo},
  {and} \bibinfo{person}{Ashraf Aboulnaga}.} \bibinfo{year}{2013}\natexlab{}.
\newblock \showarticletitle{Scalable maximum clique computation using
  mapreduce}. In \bibinfo{booktitle}{\emph{Proc. of ICDE'13}}.
  \bibinfo{pages}{74--85}.
\newblock


\bibitem[Yannakakis(1978)]%
        {STOC78:Yannakakis}
\bibfield{author}{\bibinfo{person}{Mihalis Yannakakis}.}
  \bibinfo{year}{1978}\natexlab{}.
\newblock \showarticletitle{Node- and Edge-Deletion NP-Complete Problems}. In
  \bibinfo{booktitle}{\emph{Proc. of STOC'78}}. \bibinfo{publisher}{{ACM}},
  \bibinfo{pages}{253--264}.
\newblock


\bibitem[Yu et~al\mbox{.}(2006)]%
        {Bio06:Yu}
\bibfield{author}{\bibinfo{person}{Haiyuan Yu}, \bibinfo{person}{Alberto
  Paccanaro}, \bibinfo{person}{Valery Trifonov}, {and} \bibinfo{person}{Mark
  Gerstein}.} \bibinfo{year}{2006}\natexlab{}.
\newblock \showarticletitle{Predicting interactions in protein networks by
  completing defective cliques}.
\newblock \bibinfo{journal}{\emph{Bioinform.}} \bibinfo{volume}{22},
  \bibinfo{number}{7} (\bibinfo{year}{2006}), \bibinfo{pages}{823--829}.
\newblock


\end{thebibliography}

\end{document}